\documentclass[journal,12pt,onecolumn]{IEEEtran}
\usepackage{amsbsy,amsmath,amssymb,epsfig,bbm,mathrsfs,multirow,amsthm}

\usepackage{amsmath,amsfonts}
\usepackage{algorithmic}
\usepackage{algorithm}
\usepackage{array}
\usepackage[caption=false,font=normalsize,labelfont=sf,textfont=sf]{subfig}
\usepackage{textcomp}
\usepackage{stfloats}
\usepackage{url}
\usepackage{verbatim}
\usepackage{graphicx}
\usepackage{cite}

\usepackage{bm}
\usepackage{color}
\usepackage{hyperref}
\usepackage{pdfsync} 
\usepackage{setspace}

\newtheorem{proposition}{Proposition}

\usepackage{pdfsync} 

\begin{document}
	
	%
	\title{Transmit Power Minimization for STAR-RIS Empowered Symbiotic Radio  Communications}
	%
	
	%

	\markboth{IEEE Transactions on Cognitive Communications and Networking}%
	{Shell \MakeLowercase{\textit{et al.}}: Bare Demo of IEEEtran.cls for IEEE Communications Society Journals}

	\author{Chao Zhou,~Bin Lyu,~\IEEEmembership{Member,~IEEE,}~Youhong  Feng,~\IEEEmembership{Member,~IEEE,} and\\
		Dinh Thai Hoang,~\IEEEmembership{Senior Member,~IEEE}
		\IEEEcompsocitemizethanks{\IEEEcompsocthanksitem C. Zhou and B. Lyu   are with the Key Laboratory of Ministry of Education in Broadband Wireless Communication and Sensor Network Technology, Nanjing University of Posts and Telecommunications, Nanjing 210003, China (email: zoe961992059@163.com, blyu@njupt.edu.cn).  Y. Feng is with School of Physics and Electronic information, Anhui Normal University,  Wuhu 241000, China (e-mail: yhfeng0215@126.com). D. T. Hoang is with School of Electrical and Data Engineering, University of Technology Sydney, Sydney, NSW 2007, Australia (email: hoang.dinh@uts.edu.au). 
			
		}
	}

	\maketitle

	\begin{abstract}
		In this paper, we propose a simultaneously transmitting and reflecting reconfigurable intelligent surface (STAR-RIS) empowered transmission scheme  for  symbiotic radio (SR) systems to make more flexibility for network deployment and enhance system performance. The STAR-RIS is utilized to not only beam the primary signals from the  base station (BS) towards multiple primary users on the same  side of the STAR-RIS, but also achieve the secondary transmission to the secondary users  on another side.  We consider both the broadcasting signal model and unicasting signal model at the BS. For each model, we aim for minimizing the transmit power of the BS by designing the active beamforming and simultaneous reflection and transmission coefficients under the practical phase correlation constraint. To address  the  challenge of solving the formulated problem, we propose a  block coordinate descent  based algorithm with the semidefinite relaxation, penalty dual decomposition and successive convex approximation methods, which decomposes the original problem into one sub-problem about active beamforming and the other sub-problem about simultaneous reflection and transmission coefficients, and iteratively solve them until the convergence is achieved. Numerical results  indicate that the proposed scheme can reduce up to 150.6\% transmit power compared to the backscattering device enabled scheme.
	\end{abstract}
	

	\IEEEpeerreviewmaketitle

\section{Introduction}

Recently, reconfigurable intelligent surface (RIS), as an emerging technology in 6G communications, has arouse the wide attention from academia and industry \cite{RIS-introduction1,RIS-introduction2,RIS-introduction3}. A typical hardware architecture of RIS comprises of a reflection component, a copper plate and a control circuit. 
	 Compared to  the impedance metasurface, RIS can control its phase shift from $0$ to $2\pi$  to smartly reconstruct the associated multipath \cite{360phase}. Inspired by this, RIS has been viewed as a promising solution for boosting the performance of wireless communications and widely deployed for localization, massive connectivity, edge computing,   and physical layer security\cite{RIS-introduction4}.

At present, RIS is generally classified into reflection-only RIS \cite{RIS-introduction4}   and simultaneously transmitting and reflecting RIS (STAR-RIS)\cite{STAR-RIS-intorduction1}. For the reflection-only RIS, the transmitter and receiver have to be on the same side of the RIS.  Recently, most works focus on the applications of reflection-only RIS in wireless communications \cite{RIS,HuangCW,Dai,localization1,localization2,edge-comp1,PLS-1,PLS-2}. In \cite{RIS}, the reflection-only RIS empowered  scheme was proposed to improve the downlink communication performance. In \cite{HuangCW}, the energy efficiency maximization was achieved by designing the transmit power at the base station (BS) and phase shifts at the reflection-only RIS. In \cite{Dai}, the reflection-only RIS was utilized in a cell-free network to assist the transmission from distributed BSs to multiple users.
In \cite{localization1} and \cite{localization2}, the reflection-only RIS assisted wireless localization scheme was proposed to improve the accuracy of wireless localization. Moreover, the reflection-only RIS was also applied to boost the  performance of   mobile  edge computing  systems \cite{edge-comp1} and enhance physical layer security \cite{PLS-1,PLS-2}.

Constrained by the reflection-only characteristics, the flexibility of network deployment in \cite{RIS,HuangCW,Dai,localization1,localization2,edge-comp1,PLS-1,PLS-2} cannot be guaranteed. For example, if the transmitter and receiver are deployed on different sides of the reflection-only RIS, the existence of the reflection-only RIS cannot assist the communications between the transmitter and receiver, and may be even an obstacle between them. In order to surmount this deficiency, the novel STAR-RIS \cite{STAR-RIS-introduction2}, also known as Intelligent Omni-Surface (IOS) \cite{IOS},  was proposed. In each element of the STAR-RIS, a parallel resonant LC tank and small metallic loops are adopted to afford the desired electric and magnetic surface reactance. Then, by varying the bias voltages to the integrated varactors, the electric and magnetic surface reactance of each element is adjusted to achieve the simultaneous transmission and reflection independently \cite{STAR-RIS-introduction2}. The incident signal received at each STAR-RIS element can  not only  be reflected on one side (i.e., reflection region) but also  penetrate the surface on another side (i.e., transmission region). Thus, the $360^{\circ}$  coverage  can be achieved for STAR-RIS assisted communication systems. Motivated by this superior performance, STAR-RIS has been widely investigated \cite{STAR,ChannelEsti-STAR-RIS,STAR-RIS-MIMO,STAR-RIS-secrecy1,STAR-RIS-secrecy3,IOSCell}.
 In \cite{STAR}, a STAR-RIS aided downlink communication system was studied, in which three operating protocols were proposed and the system performance in  terms of transmit power under these protocols was further compared. In \cite{ChannelEsti-STAR-RIS}, the authors designed an efficient uplink channel estimation algorithm for STAR-RIS aided systems and  demonstrated that the proposed algorithm is more suitable for the time-switching protocol. In \cite{STAR-RIS-MIMO}, a  STAR-RIS assisted (multiple-input  multiple-output, MIMO) system was investigated, in which the sum-rate maximization problem was investigated in both unicasting and broadcasting signal models.   Due to the superior ability of reconstructing the radio environment, STAR-RIS is also an efficient technique for  establishing reliable secure transmissions \cite{STAR-RIS-secrecy1,STAR-RIS-secrecy3}. In \cite{IOSCell}, a STAR-RIS aided indoor transmission scheme was proposed to reduce the inter-cell interference.

In the above works \cite{localization1,localization2,edge-comp1,PLS-1,PLS-2,STAR-RIS-introduction2,IOS,STAR,ChannelEsti-STAR-RIS,STAR-RIS-MIMO,STAR-RIS-secrecy1,STAR-RIS-secrecy3}, both reflection-only RIS and STAR-RIS were considered to assist  information transmission of existing communication systems. However, the RIS  may also need to transmit its own information, e.g., the  humidity and temperature sensed by its embedded sensors, to a desired receiver.  This novel functionality for supporting information delivery of both primary transmission (i.e., information transmission of existing communication systems) and secondary transmission (i.e., its own information transmission) is termed as a symbiotic radio (SR) paradigm \cite{LiangSurvey}, which was inspired by the concept of the biological ecosystem \cite{LiangMag}. A conventional SR system was empowered by a  backscattering device, which  transmits its own information by passively reflecting the primary signals \cite{GuoWCL,LongIoT}. To enable the practical applications of the SR paradigm, the parasitic SR and commensal SR (CSR) setups were proposed in \cite{GuoWCL} and \cite{LongIoT}, where the relationship between primary and secondary transmissions was  revealed. However, the conventional SR system generally suffers from the poor system performance due to the low communication efficiency of  backscattering devices. To break this bottleneck,
the RIS enabled SR transmission is considered as a promising technology to achieve  satisfying spectrum and energy efficiencies, and thus motivates the wide investigations in the literature \cite{RIS-inf-transfer1,RIS-inf-transfer2,RIS-inf-transfer3,QianqianZhang,RIS-SR2,RIS-SR3,RIS-SR4,RIS-SR6}.  In \cite{RIS-inf-transfer1} and \cite{RIS-inf-transfer2}, a  spatial modulation  scheme was implemented at the reflection-only RIS to enhance the  primary transmission and build the secondary transmission in the single-user single-input multiple-output (SIMO)  and  multi-user MIMO systems, respectively.  In \cite{RIS-inf-transfer3}, the authors proposed a novel reflection pattern modulation scheme for the reflection-only RIS, where a part of elements was activated for constructing a  passive beamforming for performance enhancement, and the selection of the ON-state elements was utilized to modulate its own information. In \cite{QianqianZhang}, the CSR setup  was investigated for reflection-only RIS enabled MIMO systems, in which the duration of secondary transmission is much larger than that of the primary transmission and the mutualism relationship between the two types of transmissions can be achieved. The scenario in \cite{QianqianZhang} was further extended to a general scenario with multiple primary users (PUs) in \cite{RIS-SR2}.  In  \cite{RIS-SR3} and \cite{RIS-SR4}, multiple  reflection-only RISs were utilized for achieving simultaneous primary and secondary transmissions. In \cite{RIS-SR6}, the authors analyzed the capacity of reflection-only RIS aided MIMO SR  systems, in which the reflection patterns followed a  non-uniform and discrete  distribution for capacity enhancement and additional information delivery.

In most of current works, e.g., \cite{RIS-inf-transfer1,RIS-inf-transfer2,RIS-inf-transfer3,QianqianZhang,RIS-SR2,RIS-SR3,RIS-SR4,RIS-SR6}, the RIS deployed in SR systems was limited to the reflection-only RIS. Under this setup, the PUs and secondary users (SUs) have to be deployed on the same side of the RIS, which may  not be  flexible for network deployment in practice. However, this deployment cannot be achieved by implementing the reflection-only RIS, which motivates our investigation on STAR-RIS enabled SR systems. To the best of our knowledge, it is worth noting that no works have investigated STAR-RIS enabled SR systems yet.

In this paper, we study a STAR-RIS enabled SR system for guaranteeing a more flexible network deployment. Specifically, the BS and multiple PUs are distributed on the same side of STAR-RIS, while the SUs are located on another side. The STAR-RIS is deployed to assist both the primary transmission from the BS to multiple PUs and achieve its own information delivery to the SUs. We consider two practical signal models for the primary transmission, i.e., broadcasting signal model and unicasting signal model. For the broadcasting signal model, the BS transmits common information to all PUs. While, the independent information is transmitted to each PU for the  unicasting signal model.
For each model, we aim to minimize the transmit power at the BS by  optimizing the active beamforming at the BS and the reflection and transmission coefficients at the STAR-RIS under the practical  phase  correlation constraints \cite{STARcoupledphase}. Unlike \cite{RIS-inf-transfer1,RIS-inf-transfer2,RIS-inf-transfer3,QianqianZhang,RIS-SR2,RIS-SR3,RIS-SR4,RIS-SR6}, the implementation  of STAR-RIS  results in the joint design of reflection and transmission coefficients, and thus the  methods proposed in \cite{RIS-inf-transfer1,RIS-inf-transfer2,RIS-inf-transfer3,QianqianZhang,RIS-SR2,RIS-SR3,RIS-SR4,RIS-SR6} are inappropriate for  the considered problems in this work. Moreover,  the  imperfect successive interference cancellation (SIC) is considered at the SUs, which can avoid the design mismatch caused by the perfect SIC assumption as in \cite{RIS-inf-transfer1,RIS-inf-transfer2,RIS-inf-transfer3,QianqianZhang,RIS-SR2,RIS-SR3,RIS-SR4,RIS-SR6}.
 Furthermore, different from \cite{STAR,ChannelEsti-STAR-RIS,STAR-RIS-MIMO,STAR-RIS-secrecy1,STAR-RIS-secrecy3,IOSCell}, the tradeoff between the quality of service (QoS) requirements at the PUs and SUs poses a new challenge to our work, especially for the unicasting signal model. Alternatively, our formulated problems are more hard to solve due to the combination of challenges faced by \cite{STAR,ChannelEsti-STAR-RIS,STAR-RIS-MIMO,STAR-RIS-secrecy1,STAR-RIS-secrecy3,IOSCell}  and \cite{RIS-inf-transfer1,RIS-inf-transfer2,RIS-inf-transfer3,QianqianZhang,RIS-SR2,RIS-SR3,RIS-SR4,RIS-SR6}. The main contributions of this paper are summarized as follows:

	\begin{itemize}
		\item{We propose a STAR-RIS enabled transmission scheme for SR systems to improve the system communication efficiency and achieve  the full space coverage. Specifically, the STAR-RIS is utilized not only to enhance the primary transmission by beaming the desired signals from the BS towards the PUs, but also as a secondary transmitter to realize the secondary information  transmission to the SUs by modulating  the primary signal. To meet different application requirements, we consider both broadcasting signal model and unicasting signal model for the primary transmission.  Moreover, we take into account the imperfect SIC for signal decoding at the SUs.}
		
		\item{We first investigate the transmit power minimization problem for the BS in the broadcasting signal model. To address the  challenge caused by the coupled variables, 
		 we propose a  block coordinate descent (BCD) based algorithm, which decomposes the  original problem into two sub-problems. For the sub-problem of the active beamforming optimization, the optimal solution can be obtained by transforming the original sub-problem into a convex semi-definite program (SDP) problem. For the sub-problem of reflection and transmission coefficients optimization, we propose to use the penalty dual decomposition (PDD) and successive convex approximation (SCA) methods to derive a near-optimal solution.}

		\item{We then minimize the transmit power  at the BS for the unicasting signal model, which is much challenging to solve due to the increased variables and complex constraints. To solve this problem efficiently, we extend the proposed algorithm for the broadcasting signal model. Specifically, we reformulate the active beamforming sub-problem by exploiting a diagonal matrix structure for independent primary signal transmissions and prove that the rank of the  optimal solution obtained is equal to the number of PUs. While, the PDD and SCA methods are also used to find the solution of the reflection and transmission coefficients.  }

		\item{Through numerical results, we  confirm the convergence efficiency of the proposed BCD based algorithm and show the impact of the phase correlation constraints on the design of  phase shifts. Then, we demonstrate  that the proposed STAR-RIS enabled scheme can reduce up to 150.6\% transmit power at the BS compared to the baseline schemes.}
		
	\end{itemize}

	 The rest of this paper is  organized as follows. In Section \ref{sysmodel}, we describe the STAR-RIS enabled SR system model. The transmit power minimization problem formulation and solving process for the broadcasting signal model and  the unicasting signal model are presented in Section \ref{Section3} and Section \ref{section4}, respectively.  Numerical results  are demonstrated in Section  \ref{section5} to confirm the effectiveness of the proposed scheme. Section \ref{section6} concludes  this paper.

	\section{System Model }
	\label{sysmodel}
 	As shown in Fig. \ref{SystemModel}, we consider a STAR-RIS aided SR communication system,  which consists of one BS with $N$ antennas, $K$ single-antenna PUs, $Q$ single-antenna  SUs, and one STAR-RIS with $M$ elements.  We consider that  the BS and $K$ PUs are located  in the reflection region of the STAR-RIS, while the SUs are located in the transmission region.\footnote{The system model can be extended to the scenario that the PUs and SUs are on both sides of the STAR-RIS, and the proposed algorithms in this paper are still applicable.} This is a practical scenario where  the STAR-RIS  is embedded on the surface of a building for supporting the  outdoor and indoor transmissions. For example, the PUs and SUs can be logistics devices in the outdoor and smart home controllers in the indoor, respectively. In this scenario, the STAR-RIS is used to enhance the outdoor transmission (i.e., primary transmission) from the BS to the PUs for  updating the information of recipients by exploiting its reflection characteristic, and at the same time to establish the indoor transmission  (i.e., secondary transmission) to deliver the environmental information to the SUs  for controlling smart home equipment. 
 	We consider the CSR setup for achieving  mutualistic mechanism between the primary and secondary transmissions.  Let  $T_s$ and $T_c$ be the periods of the primary and secondary signals. In the CSR setup, the secondary symbol period is a great many times that of the primary symbol period, i.e., $T_c=L T_s$, where $L \gg 1$ \cite{LongIoT}.

Let $\mathcal{K} =  \{ 1,\ldots,K \} $, $\mathcal{Q} =  \{ 1,\ldots,Q \} $ and $\mathcal{M} = \{1,\ldots,M\}$ represent  the set of PUs, the set of SUs and the set of STAR-RIS elements, respectively. We denote $\bm{h}_{p,k}^H \in \mathbb{C}^{1\times N}$, $\bm{h}_{s,q}^H \in \mathbb{C}^{1\times N}$ and $\bm {F} \in \mathbb{C}^{M\times N} $ as  the complex baseband equivalent  channels from the BS to the $k$-th PU, from the BS to the $q$-th SU, and from the BS to the STAR-RIS, respectively, where $k \in \mathcal{K}$ and $q \in \mathcal{Q}$. The complex baseband equivalent channels  from the STAR-RIS to the $k$-th PU and from the STAR-RIS to the $q$-th SU are denoted by  $\bm {g}_{p,k}^H \in \mathbb{C}^{1\times M}$  and $\bm {g}_{s,q}^H \in \mathbb{C}^{1\times M}$, respectively. Similar to \cite{STAR-RIS-secrecy1} and \cite{QianqianZhang}, we consider that  the channels follow the quasi-static fading, in which the CSI is invariable during one secondary symbol period.  Specifically, the related CSI can be obtained by using the methods proposed in \cite{ChannelEsti-STAR-RIS} and \cite{ChannelEstimation} before information transmissions.

	\begin{figure}
		\centering
		\includegraphics[width=0.4 \linewidth]{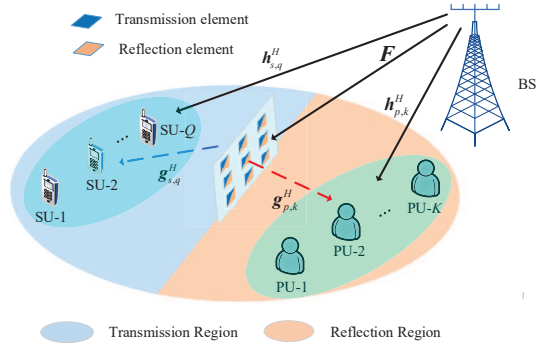}
		\caption{A STAR-RIS enabled SR communication system.}
		\label{SystemModel}
	\end{figure}

	\subsection{Energy Splitting Protocol for STAR-RIS}

We consider that the  practical energy splitting (ES) operating protocol \cite{STAR-RIS-introduction2} is adopted. Based on this protocol, each element of the STAR-RIS splits its received signal into reflected signal and transmitted signal by controlling the ES factors, i.e., the reflection coefficient and transmission coefficient \cite{STAR}, which is reconfigured by adjusting each element's electric and magnetic currents \cite{STAR-RIS-introduction2}.
	Let $v_m^r = \sqrt{\beta_m^r} e^{j\theta_m^r}$ and $v_m^t =\sqrt{\beta_m^t} e^{j\theta_m^t}$ represent the reflection coefficient and transmission coefficient of the $m$-th element of STAR-RIS, respectively.  $\sqrt{\beta_m^r} \in[0,1]$ and  $\sqrt{\beta_m^t} \in [0,1]$ are the amplitude coefficients of the $m$-th element for  reflection and transmission, respectively.  $\theta_m^r \in [0,2\pi)$ and $\theta_m^t \in [0,2\pi)$ are the corresponding phase shifts of the $m$-th element. According to \cite{STARcoupledphase} and \cite{coupledsolution}, it is known that the reflection coefficients and transmission coefficients of the STAR-RIS need to satisfy the following constraints   
	  \begin{align}
	  \label{PowerLawTwo}
	  &\beta_m^r+\beta_m^t=1,~ m \in \mathcal{M},\\
	  \label{PhaseCorrelation}
	  &\cos(\theta_m^r-\theta_m^t)=0,~ m \in \mathcal{M}. 
	  \end{align}

For the primary transmission, we consider two signal transmission models, i.e., the broadcasting signal model and the unicasting signal model. For the broadcasting signal model, the BS transmits a common data stream to all the PUs, a typical application  of which is the TV service. For the unicasting signal model,  the BS transmits an independent signal to each PU for supporting diversified service requirements.

\subsection{Broadcasting  Signal Model }
In  the broadcasting signal model, we  denote the common primary symbol as $s(l)$, where $s(l) \sim \mathcal{CN}(0,1)$, $l \in \mathcal{L}$, and $\mathcal{L} = {\left\lbrace 1,\ldots,L \right\rbrace }$ is the set of all primary symbols  in a secondary symbol period.  Considering the passive characteristic of the STAR-RIS, the  binary phase shift keying (BPSK) scheme is adopted to  modulate the secondary symbol, which is denoted by $c$ and achieved by adjusting the phase shifts.\footnote{It is worth noting that the STAR-RIS does not need to generate symbols actively for the secondary transmission \cite{RIS-inf-transfer1,RIS-inf-transfer2,RIS-inf-transfer3,QianqianZhang,RIS-SR2,RIS-SR3,RIS-SR4,RIS-SR6}.} Thus, we have $c \in {\left\lbrace 1, -1\right\rbrace }$ \cite{QianqianZhang,RIS-SR4}.
Let $\bm w$ be the active beamforming vector at the BS, and the transmitted primary signal at the BS can be expressed as $\bm w s(l)$. Then, the received signal at the $k$-th PU during the  secondary symbol period can be expressed as 
	\begin{align}
		\label{MultiUser_Rec}
		y_{b,k}^{(p)}(l)=\bm{h}_{p,k}^H \bm{w} s(l)+\bm {g}_{p,k}^H \bm{\Theta}_r \bm{F}\bm{w} s(l)c+z_{p,k}(l),~k\in \mathcal{K},
	\end{align}
	where  $z_{p,k}(l) \sim \mathcal{CN}(0,\sigma_{p,k}^2)$ is the additive Gaussian white noise (AWGN) at the $k$-th PU, and $ \bm{\Theta}_r=\text{diag}(\bm {v}_r) $ is the reflection coefficient matrix  with $\bm {v}_r=[v_1^r, \ldots, v_M^r]^T$.\footnote{According to \cite{synchronization}, the impact of the imperfect synchronization on the  detection performance in the CSR setup is slight. Thus, similar to \cite{QianqianZhang,RIS-SR2,RIS-SR3,RIS-SR4,RIS-SR6}, we consider the perfect synchronization between the direct links and STAR-RIS related links.}
	
	In the CSR setup, $s(l)$ can be regarded as being transmitted through the equivalent channel $ {\bm{h}_{p,k}^H} + \bm {g}_{p,k}^H \bm{\Theta}_r \bm Fc $ \cite{LongIoT}. From \eqref{MultiUser_Rec}, the signal-to-noise-ratio (SNR) for decoding $s(l)$ at the $k$-th PU is given by
		$\gamma_{p,k}(c)={\left|\bm{h}_{p,k}^H \bm w+\bm {g}_{p,k}^H \bm{\Theta}_r \bm F\bm w c \right| ^2}/{\sigma_{p,k}^2}, k\in \mathcal{K}$,	
	which is a function  with respect to   $c$. If $L \gg 1$, the achievable rate of $s(l)$ at the $k$-th PU in the broadcasting signal model can be approximately expressed as \cite{Tse,QianqianZhang}
	\begin{align}
		\label{multi_userk_Rate}
		R_{b,k}^{(p)} &= \mathbb{E}_{c} \left[\log_2(1+\gamma_{p,k}(c)) \right] =\frac{1}{2} \log_2(1+{\left|\bm{h}_{p,k}^H \bm w+\bm {g}_{p,k}^H \bm{\Theta}_r \bm F\bm w  \right| ^2}/{\sigma_{p,k}^2}) \nonumber \\
		&+\frac{1}{2} \log_2(1+{\left|\bm{h}_{p,k}^H \bm w-\bm {g}_{p,k}^H \bm{\Theta}_r \bm F\bm w  \right| ^2}/{\sigma_{p,k}^2}),~k \in \mathcal{K}.
	\end{align}
By exploiting the  characteristic of the STAR-RIS, the signal received by the $q$-th SU is written as 
	\begin{align}
		\label{Multi_SU_Rec}
		y_{b,q}^{(s)}(l)=\bm{h}_{s,q}^H \bm{w} s(l)+\bm {g}_{s,q}^H \bm{\Theta}_t \bm{F}\bm{w} s(l)c+z_{s,q}(l),~q\in\mathcal{Q},
	\end{align}
	where $ \bm{\Theta}_t=\text{diag}(\bm {v}_t) $ is the transmission coefficient matrix  with $\bm {v}_t=[v_1^t, \ldots, v_M^t]^T$,
	$z_{s,q}(l)\sim \mathcal{CN}(0,\sigma_{s,q}^2)$ is the AWGN at the $q$-th SU. Similar  to \eqref{multi_userk_Rate}, the achievable rate for decoding $s(l)$ at the $q$-th SU in the broadcasting signal model is approximately expressed as 
	\begin{align}
		\label{Multi_SNRandRate_SU}
		R_{b,q}^{(s)} = \frac{1}{2} \log_2(1+ {\left|\bm{h}_{s,q}^H \bm w+\bm {g}_{s,q}^H \bm{\Theta}_t \bm F\bm w  \right| ^2}/{\sigma_{s,q}^2}) 
		+\frac{1}{2} \log_2(1+{\left|\bm{h}_{s,q}^H \bm w-\bm {g}_{s,q}^H \bm{\Theta}_t \bm F\bm w  \right| ^2}/{\sigma_{s,q}^2}),~q\in\mathcal{Q}.
	\end{align}  
Under the condition that each SU can decode $s(l)$ successfully, the SIC technique is then adopted to remove the first term of \eqref{Multi_SU_Rec}, i.e., $\bm{h}_{s,q}^H \bm w s(l) $. {Considering the imperfect SIC,  the first term of \eqref{Multi_SU_Rec} related with $s(l)$ may not be removed completely. Thus, the intermediate signal in the secondary symbol period over $L$ primary symbols, denoted by $\hat{\bm{y}}_{b,q}=[\hat{{y}}_{b,q}(1),\hat{{y}}_{b,q}(2),\ldots,\hat{{y}}_{b,q}(L)]^T$,   is expressed as
	\begin{align}
		\label{Multi_SU_imtermediate_Rec}
		\hat{\bm{y}}_{b,q}=\bm {g}_{s,q}^H \bm{\Theta}_t \bm{F}\bm{w} \bm{s} c+ \hat{\bm{s}}_q+ \bm{z}_{s,q},~q\in \mathcal{Q},
	\end{align}
where $\bm{s} = [s(1),s(2),\ldots,s(L)]^T$, $\hat{\bm{s}}_q = [\hat{s}_q(1),\hat{s}_q(2),\ldots,\hat{s}_q(L)]^T$,
$\hat s_q(l)\sim \mathcal{CN}(0,\mu \left| \bm{h}_{s,q}^H \bm w  \right|^2)$ is the residual interference caused by the imperfect SIC, $\mu \in [0,1]$  is the  imperfect SIC coefficient  \cite{ImSIC}, and $\bm{z}_{s,q}= [{z}_{s,q}(1),{z}_{s,q}(2),\ldots, {z}_{s,q}(L)]^T$.
 By applying the MRC technique to decode $c$ \cite{QianqianZhang}, the corresponding  signal-to-interference-plus-noise-ratio (SINR)  at the $q$-th SU in the broadcasting signal model  is approximately formulated as
	\begin{align}
		\label{multi_gamma_c}
		\gamma_{b,q}={L \left| \bm {g}_{s,q}^H \bm{\Theta}_t \bm F\bm w\right| ^2}/{(\mu\left| \bm{h}_{s,q}^H \bm w  \right|^2 +\sigma_{s,q}^2)},~q\in \mathcal{Q},
	\end{align}
	where $L \gg 1$.}
	
	\subsection{Unicasting Signal Model}
	In the unicasting signal model, the BS simultaneously transmits independent signals to all PUs. Let $\bm w_k $ and $s_k(l)$ represent the active beamforming vector  and desired signal for $k$-th PU. The superimposed signal at the BS is expressed as $\sum_{k=1}^K \bm{w}_k s_k(l)$. Then, the received signal at the $k$-th PU from the BS is formulated as 
	\begin{align}
		\label{broadcast_User_Rec}
		y_{u,k}^{(p)}(l)&= \bm{h}_{p,k}^H \bm{w}_k s_k(l)+\bm {g}_{p,k}^H \bm{\Theta}_r \bm{F}\bm{w}_k s_k(l)c +\sum_{i=1,i\neq k}^{K}\bm {h}_{p,k}^H \bm{w}_i s_i(l)+\sum_{i=1,i\neq k}^{K}\bm {g}_{p,k}^H \bm{\Theta}_r \bm{F}\bm{w}_i s_i(l)c+z_{p,k}(l). 
	\end{align}
	The first two terms of \eqref{broadcast_User_Rec} represent the desired part for the $k$-th PU, while the third and  fourth terms correspond to the interference. From \eqref{broadcast_User_Rec}, the SINR for decoding $s_k(l)$ at the $k$-th PU   is given by
	\begin{align}
		\label{broadcast_SNR_uerk}
		\gamma_{u,k}^{(p)}(c)&=\frac{\left|\bm{h}_{p,k}^H \bm w_k+\bm {g}_{p,k}^H \bm{\Theta}_r \bm F\bm w_k c \right| ^2}{\sum_{i=1,i\neq k}^{K} \left|\bm{h}_{p,k}^H \bm w_i \right| ^2+\sum_{i=1,i\neq k}^{K} \left|\bm {g}_{p,k}^H \bm{\Theta}_r \bm F\bm w_i \right| ^2 +\sigma_{p,k}^2}, k\in \mathcal{K}.
	\end{align}
	Similar to \eqref{multi_userk_Rate}, the achievable rate of the $k$-th PU in the unicasting signal model is approximately  as 
		\begin{align}
			\label{broadcast_user_rate}
				&{R_{u,k}^{(p)}} =\mathbb{E}_{c}\left[\log_2(1+ {\gamma_{u,k}^{(p)}}(c)) \right] 
				=\frac{1}{2} \log_2(1+\frac{\left|\bm{h}_{p,k}^H \bm w_k+\bm {g}_{p,k}^H \bm{\Theta}_r \bm F\bm w_k \right| ^2}{\sum_{i=1,i\neq k}^{K} \left|\bm{h}_{p,k}^H \bm w_i \right| ^2+\sum_{i=1,i\neq k}^{K} \left|\bm {g}_{p,k}^H \bm{\Theta}_r \bm F\bm w_i \right| ^2 +\sigma_{p,k}^2}) \nonumber \\
				&~~~~~~~~~~~~~~ + \frac{1}{2}\log_2(1+\frac{\left|\bm{h}_{p,k}^H \bm w_k-\bm {g}_{p,k}^H \bm{\Theta}_r \bm F\bm w_k  \right| ^2}{\sum_{i=1,i\neq k}^{K} \left|\bm{h}_{p,k}^H \bm w_i \right| ^2+\sum_{i=1,i\neq k}^{K} \left|\bm {g}_{p,k}^H \bm{\Theta}_r \bm F\bm w_i \right| ^2 +\sigma_{p,k}^2}), ~k\in \mathcal{K}.
		\end{align} 

 The received signal at the $q$-th SU in the unicasting signal model is formulated as
	\begin{align}
		\label{broadcast_SU_Rec}
		y_{u,q}^{(s)}(l)&= \sum_{k=1}^{K}\bm{h}_{s,q}^H \bm{w}_k s_k(l)+\sum_{k=1}^{K}\bm {g}_{s,q}^H \bm{\Theta}_t \bm{F}\bm{w}_k s_k(l)c+z_{s,q}(l),~q\in \mathcal{Q}.
	\end{align}
It is known that all the SUs aims to decode the secondary signal $c$. However, before that, the decoding of $ s_k(l)$ is necessary. To achieve this goal, the SIC technique is again adopted. {According to \cite{RIS-SR4} and \cite{SICorder2}, the decoding order of the primary signals in \eqref{broadcast_SU_Rec} can be determined by their corresponding channel gains. In particular, the primary signal with the largest channel gain will be decoded first. Then, the successfully decoded signal is removed from \eqref{broadcast_SU_Rec}.  Considering the imperfect SIC,  the first term of \eqref{broadcast_SU_Rec} related with the previously decoded signal is not eliminated completely. Then, the signal with the largest channel gain in the next round is decoded by taking into account the residual interference and noise. The above process ends when all signals are decoded. Without loss of generality, we consider that $s_i(l)$ is decoded before $s_j(l)$, where $i<j$. The SINR for decoding $s_k(l)$ at the $q$-th SU in the unicasting signal model is formulated as
	\begin{align}
		\label{broadcast_SNR_SU}
		\gamma_{q,k}^{(s)}(c)&=\frac{\left|\bm{h}_{s,q}^H \bm w_k+\bm {g}_{s,q}^H \bm{\Theta}_t \bm F\bm w_k c \right| ^2}{\sum_{i=1,i\neq k}^{K} \hat \mu_i \left|\bm{h}_{s,q}^H \bm w_i \right| ^2+\sum_{i=1,i\neq k}^{K} \left|\bm {g}_{s,q}^H \bm{\Theta}_t \bm F\bm w_i \right| ^2 +\sigma_{s,q}^2}, ~k\in \mathcal{K},~q\in \mathcal{Q},	
	\end{align}
where $\hat \mu_i=\left\{ {\begin{array}{*{20}{c}}
		{\mu_i ,i < k}\\
		{1,~i > k}
\end{array}} \right.$, $\mu_i \in [0,1]$ is the imperfect SIC coefficient. Based on \eqref{broadcast_SNR_SU}, the  achievable rate of  decoding $s_k(l)$ at the $q$-th SU in the unicasting signal model is written as \eqref{broadcast_SU_rate1}, which is given by
	\begin{align}
		\label{broadcast_SU_rate1}
		R_{q,k}^{(s)} &=\mathbb{E}_{c}\left[\log_2(1+\gamma_{q,k}^{(s)}(c)) \right] =\frac{1}{2} \log_2(1+\frac{\left|\bm{h}_{s,q}^H \bm w_k+\bm {g}_{s,q}^H \bm{\Theta}_t \bm F\bm w_k  \right| ^2}{\sum_{i=1,i\neq k}^{K}(\hat \mu_i \left|\bm{h}_{s,q}^H \bm w_i \right| ^2+\left|\bm {g}_{s,q}^H \bm{\Theta}_t \bm F\bm w_i \right| ^2) +\sigma_{s,q}^2}) \nonumber \\
		&+ \frac{1}{2}\log_2(1+\frac{\left|\bm{h}_{s,q}^H \bm w_k-\bm {g}_{s,q}^H \bm{\Theta}_t \bm F\bm w_k  \right| ^2}{\sum_{i=1,i\neq k}^{K} (\hat \mu_i \left|\bm{h}_{s,q}^H \bm w_i \right| ^2+\left|\bm {g}_{s,q}^H \bm{\Theta}_t \bm F\bm w_i \right| ^2) +\sigma_{s,q}^2)}, ~k\in \mathcal{K},q\in \mathcal{Q}.
		\end{align}

	Under the condition that  each SU  can decode all the primary signals,  the remaining intermediate signal at the $q$-th SU  is recast as
	\begin{align}
		\label{broadcast_intermed_SU}
		\hat {y}_{u,q}(l)=\sum_{k=1}^{K}\bm {g}_{s,q}^H \bm{\Theta}_t \bm{F}\bm{w}_k s_k(l)c+\tilde s_{q}(l) +z_{s,q}(l),~q \in \mathcal{Q},
	\end{align}
	where $\tilde s_{q}(l) \sim \mathcal{CN}(0,\sum_{k=1}^{K}\mu_k \left| \bm{h}_{s,q}^H \bm w_k  \right|^2)$ is the residual interference caused by the imperfect SIC. Similar to \eqref{multi_gamma_c}, the  SINR for decoding $c$ at the $q$-th SU  in the unicasting signal model  is formulated as
	\begin{align}
		\label{broadcast_gamma_c}
		\gamma_{u,q}=\frac{L \sum_{k= 1}^{K}\left| \bm {g}_{s,q}^H \bm{\Theta}_t \bm F\bm w_k\right| ^2}{\sum_{k=1}^{K}\mu_k \left| \bm{h}_{s,q}^H \bm w_k  \right|^2+\sigma_{s,q}^2},~q\in \mathcal{Q}.
	\end{align} }

	\section{Transmit Power Minimization For Broadcasting Signal Model}
	\label{Section3}
	For the broadcasting signal model, we investigate the minimization problem of the transmit power at the BS by considering the joint design of  the active beamforming at the BS and the reflection and transmission coefficients at the STAR-RIS. Let $R_{b,\text{min}} $ and $\gamma_{b,\text{min}}$ represent the minimum achievable rate requirement of decoding $s_k(l)$ and the minimum SNR requirement for decoding $c$ for the  broadcasting signal model, respectively.
	Then,   the optimization problem for this model is formulated as 
	\begin{equation}\tag{$\textbf{P1}$} 
		\begin{aligned}
			\min_{\bm{w},\bm{v}_r,\bm{v}_t} ~~&  \|\bm{w} \|^2 \\ 
			\text{s.t.}~~ & \text{C1:~} R_{b,k}^{(p)} \ge {R}_{b,\text{min}},~ k\in\mathcal{K},  \\
			& \text{C2:~} R_{b,q}^{(s)} \ge {R}_{b,\text{min}},~ q\in \mathcal{Q},  \\
			& \text{C3:~} \gamma_{b,q} \ge  {\gamma}_{b,\text{min}},~ q\in \mathcal{Q},  \\
			& \text{C4:~}0 \le \beta_m^r, \beta_m^t \le 1, ~ m \in \mathcal{M}, \\
			& \text{C5:~} 0 \le{\theta_m^r, \theta_m^t} < 2 \pi, ~ m \in \mathcal{M},\\
			& \text{C6:~} \beta_m^r+\beta_m^t=1,~ m \in \mathcal{M},\\
			& \text{C7:~} \cos(\theta_m^r-\theta_m^t)=0, ~ m \in \mathcal{M}.
		\end{aligned}
	\end{equation}
 In \textbf{P1}, C1 and C2 indicate that the achievable rate of decoding $s_k(l)$ should satisfy the QoS constraint, C3 represents the QoS constraint of decoding $c$ at the SUs, C4 and C6 limit the value range of the reflection and transmission amplitude coefficients, C5 is the constraint on the  phase shifts of reflection and transmission, and C7 is the phase correlation constraint.

	It is straightforward to find that \textbf{P1} is  a non-convex problem due to the coupling of variables (i.e., $\bm w$, $\bm {v}_r$, and $\bm {v}_t $ ) in the constraints C1, C2 and C3. In addition, the coupling amplitude and phase shift constraints shown in C6 and C7 make the solving of  \textbf{P1} more challenging. Generally, it is very hard to obtain the globally optimal solution to \textbf{P1} due to its non-convexity. It is known that the BCD framework is  a promising way to address the non-convexity and can guarantee a near-optimal solution  by alternately solving several convex sub-problems derived from the originally non-convex problem. Thus, we propose  a  BCD based algorithm to solve \textbf{P1}, based on which we decompose  the variables into two blocks, i.e., $\{\bm{w} \}$ and $\{\bm {v}_r, \bm{v}_t \}$, and optimize them iteratively in an alternating manner. Specifically, the optimal active beamforming (i.e., $\{\bm{w} \}$) can be found by transforming the original problem into a semi-definite program (SDP). While, for the sub-problem of optimizing the reflection and transmission coefficients (i.e., $\{\bm {v}_r, \bm{v}_t \}$),  the PDD framework \cite{coupledsolution} with the SCA technique is utilized. Based on this framework, the sub-problem is transformed into an augmented Lagrangian problem by considering the penalty term, for which the variables and introduced penalty factor are also iteratively optimized until the convergence is achieved.
	
	\subsection{Active Beamforming Optimization}
	In this subsection, we  first focus on the optimization of active beamforming with  $\bm {v}_r$ and $\bm {v}_t$ being fixed.  The sub-problem for optimizing the active beamforming is given by
	\begin{equation}\tag{$\textbf{P2}$} 
		\begin{aligned}
			\min_{\bm{w}} ~~&  \|\bm{w} \|^2 \\ 
			\text{s.t.}~~ & \text{C1},~ \text{C2},~ \text{C3}.
		\end{aligned}
	\end{equation}


To address the non-convexity of \textbf{P2},  we first  transform its objective function  into a tractable form by letting $\bm W= \bm w \bm {w}^H\in \mathcal{C}^{N \times N}$, where $\bm W$  is a semi-definite matrix with $\text{Rank} (\bm{W})  = 1$. To make the constraints concise, let $\bm {e}_{1,k}^H=\bm{h}_{p,k}^H +\bm {g}_{p,k}^H \bm{\Theta}_r \bm F $, $\bm {e}_{2,k}^H=\bm{h}_{p,k}^H -\bm {g}_{p,k}^H \bm{\Theta}_r \bm F $, $\overline{ \bm {e}}_{1,q}^H=\bm{h}_{s,q}^H +\bm {g}_{s,q}^H \bm{\Theta}_t \bm F$, and $\overline{\bm {e}}_{2,q}^H=\bm{h}_{s,q}^H -\bm {g}_{s,q}^H \bm{\Theta}_t \bm F$. Based on these auxiliary variables, we further let  $\bm E_{1,k}=\bm {e}_{1,k}\bm {e}_{1,k}^H$, $\bm E_{2,k}=\bm {e}_{2,k}\bm {e}_{2,k}^H$, $\overline{\bm E}_{1,q}=\overline{\bm{e}}_{1,q}\overline{\bm{e}}_{1,q}^H$, $\overline{\bm E}_{2,q}=\overline{\bm{e}}_{2,q} \overline{\bm{e}}_{2,q}^H$, $\bm R_q =(\bm {g}_{s,q}^H \bm{\Theta}_t \bm F)^H(\bm {g}_{s,q}^H \bm{\Theta}_t \bm F)$, and $\bm H_q=\bm{h}_{s,q} \bm{h}_{s,q}^H$. By using the auxiliary variables and applying the SDR technique, \textbf{P2} can be transformed into a convex  SDP as 
	\begin{equation}\tag{${\textbf{P2.1}} $} 
		\begin{aligned}
			\min_{\bm{W}} ~~&  \text{Tr}(\bm W) \\ 
			\text{s.t.}~~ 
			\label{P2C1-3-transform}
		 &{\overline{\text{C1}}}:\sum_{i=1}^{2}\log_2(1+\frac{\text{Tr}(\bm E_{i,k}\bm W )}{\sigma_{p,k}^2})\ge 2{R}_{b,\min}, ~ k\in\mathcal{K}, \nonumber \\
		& \overline{\text{C2}}:\sum_{i=1}^{2}\log_2(1+\frac{\text{Tr}(\overline{\bm E}_{i,q}\bm W )}{\sigma_{s,q}^2})\ge 2{R}_{b,\min}, ~ q\in\mathcal{Q},\nonumber \\
		& \overline{\text{C3}}:\frac{L\text{Tr}(\bm R_q\bm{W} )}{\mu \text{Tr}(\bm{H}_q \bm{W})+\sigma_{s,q}^2}\ge { \gamma}_{b,\min},  ~ q\in\mathcal{Q}, \nonumber \\
		 &\text{C8:~} \bm W \succeq {0}.
		\end{aligned}
	\end{equation}
	\textbf{P2.1} can be solved by using the CVX \cite{CVX}. According to \cite {rank-one}, we can prove that  the optimal solution to \textbf{P2.1}, denoted by $\bm{W}^*$, is always rank-one. From this observation, we find that the optimal solution to \textbf{P2} can be obtained by solving \textbf{P2.1}. By applying the singular value decomposition (SVD) of $\bm{W}^*$, we can finally derive the optimal active beamforming $\bm w^*$ to \textbf{P2}. Moreover, this observation implies that the BS beams the transmitted signal to all the PUs towards a desired direction, which can guarantee the QoS constraints at the PUs and SUs and results in the minimum transmit power consumption at the BS.

	\subsection{Reflection and Transmission Coefficients Optimization }
	\label{multi_sectionB}
	After obtaining the optimal active beamforming $\bm w^*$, the sub-problem for optimizing the reflection and transmission coefficients derived from \textbf{P1} is expressed as 
	\begin{equation}\tag{$\textbf{P3}$} 
		\begin{aligned}
			&\text{Find} ~\{\bm{v}_r,\bm{v}_t\} \\ 
			& \text{s.t.}~~  \text{C1}-\text{C7}.
		\end{aligned}
	\end{equation}
	It is difficult to solve \textbf{P3} due to the coupled amplitude and phase shift constraints. 
According to \cite{coupledsolution}, auxiliary variables $ \tilde{\bm {v}}_r$ and $ \tilde{\bm {v}}_t$ can be introduced to tackle this challenge, where $ \tilde{\bm {v}}_r =[\tilde{v}_1^r,\ldots,\tilde{v}_M^r]^T$, $ \tilde{\bm {v}}_t =[\tilde{v}_1^t,\ldots,\tilde{v}_M^t]^T$, $ \tilde{v}_m^r =\sqrt{ \tilde {\beta}_m^r}e^{j \tilde{\theta}_m^r}$, and $ \tilde{v}_m^t =\sqrt{ \tilde {\beta}_m^t}e^{j \tilde{\theta}_m^t}, m \in \mathcal{M}$.   Then,  \textbf{P3} can be formulated as
	\begin{equation}\tag{$\textbf{P3.1}$} 
		\begin{aligned}
			&\text{Find} ~\{\bm{v}_r,\bm{v}_t,\tilde{ \bm{v}}_r,\tilde{\bm{v}}_t\} \\ 
			 \text{s.t.}~~& \text{C1}-\text{C6},\\
			& \text{C9:~} \tilde{ \bm{v}}_r=\bm{v}_r,\tilde{\bm{v}}_t=\bm{v}_t,\\
			& \text{C10:~} \tilde {\beta}_m^r+\tilde {\beta}_m^t=1,~m\in \mathcal{M},\\
			& \text{C11:~} \cos(\tilde{\theta}_m^r-\tilde{\theta}_m^t)=0,~m\in \mathcal{M}.
		\end{aligned}
	\end{equation}
In \textbf{P3.1}, the constraints C9-C11 are introduced to guarantee that the coupling phase shifts are satisfied.   To handle the equality constraint C9, the PDD framework \cite{PDD_framework} is utilized. Specifically, the penalty term associated with C9 is added as the objective function, based on which  \textbf{P3.1} is transformed as
	\begin{equation}\tag{$\textbf{P3.2}$} 
		\begin{aligned}
			\min_{{\bm{v}_r,\bm{v}_t,\tilde{ \bm{v}}_r,\tilde{\bm{v}}_t}} ~~& \frac{1}{2\rho}\sum_{\iota \in\{r,t\}}\|\tilde{ \bm{v}}_{\iota}-\bm{v}_{\iota}+\rho\bm{\lambda_{\iota}}   \| ^2 \\ 
			\text{s.t.}~~& \text{C1}-\text{C6},\\
			& \text{C10:~} \tilde {\beta}_m^r+\tilde {\beta}_m^t=1,~m\in \mathcal{M},\\
			& \text{C11:~} \cos(\tilde{\theta}_m^r-\tilde{\theta}_m^t)=0,~m\in \mathcal{M},
		\end{aligned}
	\end{equation}
	where $\rho$ is a nonnegative penalty factor,  $\bm{\lambda}_{\iota}$  represents the Lagrangian dual variables related to C9, and  $\iota \in \{r,t\}$. It is known that if $\rho \to 0$, we set the objective function to be zero to guarantee that C9 is satisfied. According to \cite{PDD_framework}, the alternating update of  $\{{\bm{v}_r,\bm{v}_t,\tilde{ \bm{v}}_r,\tilde{\bm{v}}_t}\}$, $\bm{\lambda}_{\iota}$, and $\rho$ will lead to the Karush–Kuhn–Tucker (KKT) conditions, based on which the optimal solution can be obtained. For the optimization of $\{{\bm{v}_r,\bm{v}_t,\tilde{ \bm{v}}_r,\tilde{\bm{v}}_t}\}$, \textbf{P3.2} can be divided into two sub-problems  associated with   $\{\bm{v}_r, \bm{v}_t \}$ and $ \{\tilde{ \bm{v}}_r, \tilde{ \bm{v}}_t \}$, respectively.

	\subsubsection[${v}_r, {v}_t $]{Optimization of $\{ \bm {v}_r,  \bm {v}_t \}$}
	Given $ \tilde{ \bm{v}}_r $, $ \tilde{ \bm{v}}_t $,  $\rho$, and $\bm{\lambda}_{\iota},\iota \in {\{r,t\}}$, the sub-problem of  optimizing $\{\bm{v}_r, \bm{v}_t \}$ derived from \textbf{P3.2} is written as
	\begin{equation}\tag{$\textbf{P3.3}$} 
		\begin{aligned}
			\min_{\bm{v}_r,\bm{v}_t} ~~& \frac{1}{2\rho}\sum_{\iota \in\{r,t\}}\|\tilde{ \bm{v}}_{\iota}-\bm{v}_{\iota}+\rho\bm{\lambda_{\iota}}   \| ^2 \\ 
			\text{s.t.}~~& \text{C1}-\text{C6}.
		\end{aligned}
	\end{equation}

To make the expressions concise, we let $\bm {g}_{p,k}^H \bm{\Theta}_r \bm F\bm w =\bm {g}_{p,k}^H\text{diag}(\bm F\bm w)\bm {v}_r$ and  {$\bm {g}_{s,q}^H \bm{\Theta}_t \bm F\bm w =\bm {g}_{s,q}^H\text{diag}(\bm F\bm w)\bm {v}_t$}. It is known that $R_{b,k}^{(p)}$ in the constraint C1 with respect to $\bm {v}_r$ (i.e., $\bm{\Theta}_r$) is non-convex. To deal with this issue, we approximate $R_{b,k}^{(p)}$ into a convex function by exploiting its lower-bound, which  is approximated as 
	\begin{align}
		\label{multi_SCA}
		R_{b,k}^{(p)}& \ge \frac{1}{2}\log_2 (1+ \frac{d_{p,k} d_{p,k}^H +2 \mathcal{R}\{d_{p,k}^H \bm {r}_{p,k}^H \bm {v}_r\}+|\bm {r}_{p,k}^H \bm {v}_r^{(x) }|^2 }{\sigma_{p,k}^2}	+\frac{2 \mathcal{R} \{  (\bm {v}_r^{(x)})^H  \bm {r}_{p,k}\bm {r}_{p,k}^H (\bm {v}_r-\bm {v}_r^{(x)})\}}{\sigma_{p,k}^2} ) \\
		&+ \frac{1}{2}\log_2 (1+ \frac{d_{p,k} d_{p,k}^H -2 \mathcal{R}\{d_{p,k}^H \bm {r}_{p,k}^H \bm {v}_r\}+|\bm {r}_{p,k}^H \bm {v}_r^{(x) }|^2 }{\sigma_{p,k}^2}	+\frac{2 \mathcal{R} \{ (\bm {v}_r^{(x) })^H\bm {r}_{p,k}\bm {r}_{p,k}^H (\bm {v}_r-\bm {v}_r^{(x)})\}}{\sigma_{p,k}^2} )  \triangleq \overline{R}_{b,k}^{(p)},\nonumber
	\end{align} 
	where $d_{p,k}=\bm{h}_{p,k}^H\bm{w}$, $\bm {r}_{p,k}^H=\bm {g}_{p,k}^H\text{diag}(\bm F\bm w)$, and {$\bm {v}_r^{(x)}$} is a feasible point of $\bm{\bm {v}_r}$ in the {$x$-th} iteration of the SCA method. It can be found that $\overline{R}_{b,k}^{(p)}$ is a concave function of $\bm{\bm {v}_r}$.

	Similarly, the lower-bound expressions of {$R_{b,q}^{(s)}$ and $\gamma_{b,q}$} are respectively given by 
	\begin{align}
		\label{multi_SCA2}
		{R}_{b,q}^{(s)}& \ge \frac{1}{2}\log_2(1+ \frac{d_{s,q} d_{s,q}^H +2 \mathcal{R}\{d_{s,q}^H \bm {r}_{s,q}^H \bm {v}_t\}+|\bm {r}_{s,q}^H \bm {v}_t^{(x) }|^2 }{\sigma_{s,q}^2}+\frac{2 \mathcal{R} \{ (\bm {v}_t^{(x) })^H\bm {r}_{s,q}\bm {r}_{s,q}^H (\bm {v}_t-\bm {v}_t^{(x)})\}}{\sigma_{s,q}^2}) \\
		&+ \frac{1}{2}\log_2(1+ \frac{d_{s,q} d_{s,q}^H -2 \mathcal{R}\{d_{s,q}^H \bm {r}_{s,q}^H \bm {v}_t\}+|\bm {r}_{s,q}^H \bm {v}_t^{(x) }|^2 }{\sigma_{s,q}^2}  + \frac{2 \mathcal{R} \{ (\bm {v}_t^{(x) })^H\bm {r}_{s,q}\bm {r}_{s,q}^H (\bm {v}_t-\bm {v}_t^{(x)})\}}{\sigma_{s,q}^2})  \triangleq \overline{R}_{b,q}^{(s)}, \nonumber \\
	\gamma_{b,q}	&  \ge \frac{L}{\chi_q +\sigma_{s,q}^2} |\bm {r}_{s,q}^H \bm {v}_t^{(x)} |^2  + 2 \frac{L}{\chi_q+\sigma_{s,q}^2} \mathcal{R} \{({\bm {v}_t^{(x)}}) ^H \bm {r}_{s,q}\bm {r}_{s,q}^H (\bm {v}_t-\bm {v}_t^{(x)}) \}  \triangleq \overline{\gamma}_{b,q},
	\end{align} 
	where $d_{s,q}=\bm{h}_{s,q}^H\bm{w}$,  $\bm {r}_{s,q}^H=\bm {g}_{s,q}^H\text{diag}(\bm F\bm w)$, $\chi_q=\mu \text{Tr}(\bm H_q W)$, and {$\bm {v}_t^{(x)}$} is a feasible point of $\bm{\bm {v}_t}$ in the {$x$}-th iteration of the SCA method.

	Moreover, in order to improve the convergence performance, a residual variable vector, denoted by {$ \bm \alpha=[\alpha_1,\ldots,\alpha_{K+2Q}]^T$}, is introduced. Then, \textbf{P3.3} can be  rewritten as 
	 \begin{equation}\tag{$\textbf{P3.4}$} 
	 	\begin{aligned}
	 		\min_{\bm{v}_r,\bm{v}_t,\bm \alpha} ~~& -\sum_{k=1}^{K+2Q} \alpha_k +\frac{1}{2\rho}\sum_{\iota \in\{r,t\}}\|\tilde{ \bm{v}}_{\iota}-\bm{v}_{\iota}+\rho\bm{\lambda_{\iota}}   \| ^2 \\ 
	 		\text{s.t.}~~& {\overline{\overline{\text{C1}}}}: \overline{R}_{b,k}^{(p)}  \ge R_{b, \text{min}}+\alpha_k, k\in \mathcal{K},\\
	 		& \overline{\overline {\text{C2} }}:\overline{R}_{b,q}^{(s)} \ge R_{b,\text{min}}+\alpha_{K+q}, q\in \mathcal{Q}, \\
	 		&\overline {\overline {\text{C3} }}:\overline{\gamma}_{b,q} \ge \gamma_{b,\text{min}} + {\alpha_{K+Q+q}}, q\in \mathcal{Q}, \\
	 		& {{\text{C12} }:|v_m^r|^2 +|v_m^t|^2 = 1,m\in \mathcal{M}.}
	 	\end{aligned}
	 \end{equation}
	As \textbf{P3.4} is convex,  the CVX tool can be utilized to solve it.

	 \subsubsection[$ \{\tilde{ {v}}_r, \tilde{{v}}_t \}$]{Optimization of $ \{\tilde{ \bm{v}}_r, \tilde{ \bm{v}}_t \}$}
	 \label{dual_optimizae}
	 With the fixed $\bm {v}_r$, $\bm {v}_t$, $\rho$, and $\bm{\lambda}_{\iota},\iota \in {\{r,t\}}$, the sub-problem for optimizing $ \{\tilde{ \bm{v}}_r, \tilde{ \bm{v}}_t \}$ is formulated as 
	 \begin{equation}\tag{$\textbf{P3.5}$} 
	 	\begin{aligned}
	 		\min_{\tilde{ \bm{v}}_r,\tilde{\bm{v}}_t} ~~&\sum_{\iota \in\{r,t\}} \|\tilde{ \bm{v}}_{\iota}-\bm{v}_{\iota}+\rho\bm{\lambda_{\iota}}   \| ^2 \\ 
	 		\text{s.t.}~~& \text{C10:~} \tilde {\beta}_m^r+\tilde {\beta}_m^t=1,~m\in \mathcal{M},\\
	 		& \text{C11:~} \cos(\tilde{\theta}_m^r-\tilde{\theta}_m^t)=0,~m\in \mathcal{M}.
	 	\end{aligned}
	 \end{equation}
	It is  found that \textbf{P3.5} is non-convex due to the coupled amplitude and phase shift constraints. However, this limitation can be overcome by optimizing the amplitude coefficients and phase shifts in an alternating manner. Denote $\tilde { \bm {\beta}}_{\iota} =[\sqrt{\tilde{\beta}_{1}^{\iota}},\ldots,\sqrt{\tilde{\beta}_{M}^{\iota}}]^T$ and $ \tilde { \bm {\varphi}}_{\iota} =[ e^{j\tilde{\theta}_{1}^{\iota}},\ldots, e^{j\tilde{\theta}_{M}^{\iota}}]^T$, where $\tilde { \bm {\beta}}_{\iota}$ is the amplitude vector, $ \tilde { \bm {\varphi}}_{\iota}$ is the phase shift vector, and $\iota \in \{r,t\}$. Then, we have $\tilde{\bm{v}}_{\iota}=\text{diag}(\tilde { \bm {\beta}}_{\iota})\tilde { \bm {\varphi}}_{\iota}=\text{diag}(\tilde { \bm {\varphi}}_{\iota})\tilde { \bm {\beta}}_{\iota}, \iota \in \{r,t\} $.
	Under the condition that the constraint C10 holds, we first reformulate the objective function as 
	\begin{align}
	\label{Objective}
		&\sum_{\iota \in\{r,t\}} \|\tilde{ \bm{v}}_{\iota}-\bm{v}_{\iota}+\rho\bm{\lambda}_{\iota}   \| ^2  = \sum_{\iota \in\{r,t\}} \|\tilde{ \bm{v}}_{\iota}+\bm{\varphi}_{\iota}   \| ^2 
		= \sum_{\iota \in\{r,t\}}(\tilde{ \bm{v}}_{\iota}^H \tilde{ \bm{v}}_{\iota} +\bm{\varphi}_{\iota}^H \bm{\varphi}_{\iota} +2\mathcal{R} \{ \bm{\varphi}_{\iota}^H \tilde{ \bm{v}}_{\iota} \} ) \nonumber \\
		&~~~=N +\sum_{\iota \in\{r,t\}} \bm{\varphi}_{\iota}^H \bm{\varphi}_{\iota} +\sum_{\iota \in\{r,t\}}  2\mathcal{R} \{ \bm{\varphi}_{\iota}^H \text{diag}(\tilde { \bm {\beta}}_{\iota})\tilde { \bm {\varphi}}_{\iota} \},
	\end{align}
where $ \bm{\varphi}_{\iota}=\rho \bm{\lambda}_{\iota} - \bm v_{\iota} $. It is observed that only the third term of \eqref{Objective} is with respect to $\tilde{\bm {\beta}}_{\iota}$ and  $ \tilde { \bm {\varphi}}_{\iota}$.  Thus, \textbf{P3.5} can be reduced as {
	 \begin{equation}\tag{$\textbf{P3.6}$} 
	 	\begin{aligned}
	 		\min_{ \tilde { \bm {\beta}}_{r},\tilde { \bm {\beta}}_{t}, \tilde {\bm {\varphi}}_{r},\tilde { \bm {\varphi}}_{t}} ~~& \sum_{\iota \in\{r,t\}}  \mathcal{R} \{ \bm{\varphi}_{\iota}^H \text{diag}(\tilde { \bm {\beta}}_{\iota})\tilde { \bm {\varphi}}_{\iota} \}   \\ 
	 		\text{s.t.}~~& \text{C10},~\text{C11}.
	 	\end{aligned}
	 \end{equation} }
	 According to \cite{coupledsolution}, the solution to  \textbf{P3.6} can be derived from Proposition \ref{coupling_phase_closedform} and Proposition \ref{coupling_amplitude_closedform}.

	 \begin{proposition}
	 	\label{coupling_phase_closedform}
	 	 Let $ \bm {\psi}_{\iota} =\text{diag}(\tilde { \bm {\beta}}_{\iota}^H)\bm{\varphi}_{\iota}=[\psi_{1}^{\iota},\ldots,\psi_{M}^{\iota}]^H, \iota \in \{r,t\}$. Given the amplitude vectors,  the optimal phase shifts for $m \in \mathcal{M}$  are given by
	 	\begin{align} 
	 		{(\tilde{\theta}_{m}^{r})}^*&=\pi -\text{arg}(\psi_{m}^{r}+j\psi_{m}^{t}),{(\tilde{\theta}_{m}^{t})}^*= \tilde{\theta}_{m}^{r}+\frac{1}{2}\pi, \\
	 		{(\tilde{\theta}_{m}^{r})}^*&=\pi -\text{arg}(\psi_{m}^{r}-j\psi_{m}^{t}),{(\tilde{\theta}_{m}^{t})}^*= \tilde{\theta}_{m}^{r}-\frac{1}{2}\pi.	
	 	\end{align} 
	 	As there are two solutions for the optimal phase shifts, one of which leading to a smaller value of the objective function is finally chosen.
	 	\end{proposition}

 	\begin{proposition}
 		\label{coupling_amplitude_closedform}
 	  	Let $\overline{\bm {\psi}} _{\iota} =\text{diag}(\tilde { \bm {\varphi}}_{\iota}^H)\bm{\varphi}_{\iota}=[\overline{\psi}_{1}^{\iota},\ldots,\overline{\psi}_{M}^{\iota}]^H, \iota \in \{r,t\}$, $a_m^r=\mathcal{R}\{\overline{\psi}_{m}^{r}\}$, $a_m^t=\mathcal{R}\{\overline{\psi}_{m}^{t}\}$, and $\xi_m=\text{sgn}(a_m^t) \times \text{arccos}(\frac{a_m^r}{\sqrt{{(a_m^r)}^2+{(a_m^t)}^2}})$. Given the phase shifts,  the optimal amplitude coefficients for $m \in \mathcal{M}$ are given by
 	 \begin{align} 
 	 	\label{opt_amplitude}
 			\sqrt{({\tilde{\beta}_{m}^{r}})^*}=\text{sin}(\omega_m),
 			\sqrt{({\tilde{\beta}_{m}^{t}})^*}=\text{cos}(\omega_m),
 	 \end{align}  	
 	  where
 	  \begin{equation}
 	  	\label{solution_amplitude}
 	  	\omega_m=\left\{
 	  	\begin{aligned}
 	  		&-\frac{1}{2}\pi-\xi_m,  \quad  &\text{if}~~ \xi_m \in [-\pi,-\frac{1}{2}\pi),\\
 	  		&~0, \quad  &\text{if}~~ \xi_m \in [-\frac{1}{2}\pi,\frac{1}{4}\pi),\\
 	  		&~\frac{1}{2}\pi, \quad ~~~~~~&\text{else}. \quad~~~~~~~~~~~~~~~
 	  	\end{aligned}
 	  	\right
 	  	.
 	  \end{equation}	 	  	
 	\end{proposition}	

 	The proof of Proposition \ref{coupling_phase_closedform} and Proposition \ref{coupling_amplitude_closedform} can be found in \cite{coupledsolution} and thus omitted here for simplicity. 
 	By alternating updating the optimal amplitude coefficients and phase shifts in Proposition \ref{coupling_phase_closedform} and  Proposition \ref{coupling_amplitude_closedform}, we can finally obtain a near-optimal solution of \textbf{P3.6}. 
 	
   \subsection[$\lambda_{\iota}$]{Updating ${\bm {\lambda}}_{\iota}$ and $\rho$}
   With the obtained $\{{\bm{v}_r,\bm{v}_t,\tilde{ \bm{v}}_r,\tilde{\bm{v}}_t}\}$, we proceed to update $\bm{\lambda}_{\iota}$ and $\rho$. Specifically, $\bm{\lambda}_{\iota}$ is updated by 
   \begin{align}
   \label{VUpdate}
\bm {\lambda_{\iota}}=\bm {\lambda_{\iota}}+\frac{1}{\rho}(\tilde{ \bm{v}}_{\iota}- \bm{v}_{\iota}) , {\iota} \in \{r,t\}. 
   \end{align}
 From \eqref{VUpdate}, it is found that the update of $\bm{\lambda_{\iota}}$  is controlled by the  violation with respect to the constraint C9 and the plenty factor.
Then, $\rho$ is updated by 
\begin{align}
\label{RhoUpdate}
\rho = \overline{c} \rho, 
\end{align}
where $0< \overline{c} < 1$ denotes the step size. From \eqref{RhoUpdate}, it is known that the plenty factor is  gradually  decreased with the number of iterations as to balance the convergence speed and solving accuracy. 
   
   \subsection{Algorithm Summary and Analysis}
The  proposed BCD-based algorithm for solving \textbf{P1} is summarized in Algorithm \ref{AlgorithmA}.  It is found that the transmit power is non-increasing after each iteration, which can be proved in a similar way as that in \cite{Convergence}.
   Thus, the value of the objective function in \textbf{P1} will converge  to a stationary point. Then, we analyze the computational complexity of  Algorithm \ref{AlgorithmA}. According to  \cite{complexity}, the complexities of solving \textbf{P2.1}  and \textbf{P3.4} are 
    $\mathcal{O}(\sqrt{K+2Q+N} \log(\frac{1}{\epsilon})\cdot(n_1(K+2Q+N^3)+n_1^2(K+2Q+N^2)+n_1^3  ) )$ and $\mathcal{O}(\sqrt{2K+4Q+2M} \log(\frac{1}{\epsilon})\cdot(n_2(2K+4Q)+n_2^2(2K+4Q)+n_2(4M)+n_2^3  ) )$, respectively, where $n_1=N^2$ and $n_2=2M+K+2Q$. Therefore, the  overall computational complexity of Algorithm \ref{AlgorithmA} is $\mathcal{O}(I_1 I_2\sqrt{K+2Q+N} \log(\frac{1}{\epsilon})\cdot(n_1(K+2Q+N^3)+n_1^2(K+2Q+N^2)+n_1^3  ) 
    + I_1 I_2 I_3\sqrt{2K+4Q+2M} \log(\frac{1}{\epsilon})\cdot(n_2(2K+4Q)+n_2^2(2K+4Q)+n_2(4M)+n_2^3  ) ) $, where $I_1$ is the outer iteration number for updating  $\bm \lambda_{\iota}$ and  $\rho$, $I_2$ is the inner iteration number for the proposed BCD algorithm, $I_3$ is the iteration number for solving \textbf{P3.4} by using the SCA method, respectively.
	 \begin{algorithm}
	 	\caption{ The BCD based Algorithm for Problem \textbf{P1}}
	 	\label{AlgorithmA}
	 	\begin{algorithmic}[1]  
	 		\STATE {Initialization: transmission and reflection beamforming vector $\bm{v}_{\iota}$, 
	 			  convergence tolerance $ \epsilon$,
	 			  outer iteration index $t_1=0$, 
	 			  inner iteration index $t_2=0$,
	 			  {SCA iteration index $x$}, the maximum violation value $\eta=10$, $\overline c=0.1$, $\bm \lambda_{\iota} =\bm 0 $, and $\rho=1$.} 
	 		\REPEAT 
	 		\STATE {$t_1=t_1+1$.}
	 		\REPEAT    
	 		\STATE {$t_2=t_2+1$ and {$x=0$}.}
	 		\STATE{Update $\bm{w}^{(t_1,t_2)}$ by solving \textbf{P2.1}.} \label{step6}
	 		\REPEAT
	 		\STATE {$x=x+1$.}
	 		\STATE{Obtain ${\bm {v}_r^{(x)}}$ and ${\bm {v}_t^{(x)}}$ by solving \textbf{P3.4}. }
	 		\UNTIL{the decrease of the objective function  value in \textbf{P3.4} is below $ \epsilon$.}  
	 		\STATE{Update $\bm{v}_r^{(t_1,t_2)}={\bm {v}_r^{(x)}}$ and $\bm{v}_t^{(t_1,t_2)}={\bm {v}_t^{(x)}}$}.
	 		\STATE{Update $\tilde{\bm \varphi}_{r}$ and $\tilde{\bm \varphi}_{t}$ by Proposition \ref{coupling_phase_closedform}.}
	 		\STATE{Update $\tilde{\bm \beta}_{r}$ and $\tilde{\bm \beta}_{t}$ by Proposition \ref{coupling_amplitude_closedform}.}    
	 		\UNTIL{the decrease of the objective function value in \textbf{P1} is below $\epsilon$.}
	 		\STATE{Update the constraint violation value $\delta$ by implementing $\delta=\text{max}\{\|\tilde{ \bm{v}}_r- \bm{v}_r \|_{\infty}, \|\tilde{ \bm{v}}_t- \bm{v}_t \|_{\infty} \}$.}
	 		\IF{$\delta \le \eta$} 
	 		\STATE{$\bm {\lambda_{\iota}}=\bm {\lambda_{\iota}}+\frac{1}{\rho}(\tilde{ \bm{v}}_{\iota}- \bm{v}_{\iota}) $, ${\iota} \in \{r,t\}$. }
	 		\ELSE{}
	 		\STATE{$\rho=\overline c\rho$.}
	 		\ENDIF
	 		\STATE{$ \eta=0.9\delta $.}
	 		\UNTIL{$\delta$ is below $\epsilon$.}
	 	\end{algorithmic}
	 \end{algorithm}

	\section{Transmit Power Minimization For Unicasting Signal Model}
	\label{section4}
	In this section, we continue to minimize  the transmit power at the BS in the unicasting signal model. To achieve this goal, we also focus on designing the active beamforming at the BS and the reflection and transmission coefficients at the STAR-RIS.  The optimization problem in the  unicasting signal model is formulated as 
	 \begin{equation}\tag{$\textbf{P4}$} 
	 	\begin{aligned}
	 		\min_{\bm{w}_k,\bm{v}_r,\bm{v}_t} ~~&  \sum_{k=1}^{K} \|\bm{w}_k \|^2 \\ 
	 		\text{s.t.}~~ & \text{C13:~} R_{u,k}^{(p)} \ge {R}_{u,\text{min}},~ k\in\mathcal{K},  \\
	 		& \text{C14:~} R_{q,k}^{(s)} \ge {R}_{u,\text{min}},~ k\in\mathcal{K},~ q\in\mathcal{Q},  \\
	 		& \text{C15:~} \gamma_{u,q} \ge  {\gamma}_{u,\text{min}}, ~ q\in\mathcal{Q}, \\
	 		& \text{C4},~\text{C5},~\text{C6},~\text{C7}.
	 	\end{aligned}
	 \end{equation}
	 In \textbf{P4}, ${R}_{u,\text{min}}$ and ${\gamma}_{u,\text{min}}$ denote the minimum achievable rate and SNR requirements for decoding $s_k(l)$ and $c$ for the unicasting signal model, respectively. It is worth noting that compared to \textbf{P1}, it is more challenging to solve \textbf{P4} due to the following aspects. Firstly,
	 unlike the  broadcasting signal model, the constraint C14 for $k\in \mathcal{K}$ needs to be satisfied to ensure that all the primary signals are successfully decoded at the {$q$-th} SU. Secondly, the mutual interference between the independent primary signals makes the constraints C13 and C14 more challenging. To deal with the above difficulties, we extend the proposed BCD-based algorithm in Section \ref{Section3} to solve \textbf{P4}, based on which $\{\bm{w} \}$ and $\{\bm {v}_r, \bm{v}_t \}$ are also optimized in an alternating manner. Specifically, for the sub-problem of optimizing $\{\bm{w} \}$, we exploit the diagonal matrix structure for independent primary signal transmissions and transform the original sub-problem into a convex SDP.  For the optimization of $\{\bm {v}_r, \bm{v}_t \}$, the PDD framework with the SCA technique is also applied.
	 
	\subsection{Active Beamforming Optimization}
	\label{subsectionA}
	In this subsection, we find the optimal design of the active beamforming vectors at the BS with the fixed  coefficients of the STAR-RIS by solving the following problem
	\begin{equation}\tag{$\textbf{P5}$} 
		\begin{aligned}
			\min_{\bm{w}_k} ~~& \sum_{k=1}^{K} \|\bm{w}_k \|^2 \\ 
			\text{s.t.}~~ & \text{C13},~ \text{C14},~ \text{C15}.
		\end{aligned}
	\end{equation}
	It is worth noting that the joint optimization of multiple active beamforming matrices at the BS is challenging and the methods proposed in \cite{RIS-inf-transfer1,RIS-inf-transfer2,RIS-inf-transfer3,QianqianZhang,RIS-SR2,RIS-SR3,RIS-SR4,RIS-SR6} and \cite{coupledsolution} cannot be utilized here directly. To solve \textbf{P5} efficiently, we also exploit the SDR method to solve \textbf{P5}. Let $\bm {W}_k=\bm {w}_k\bm {w}_k^H \in \mathbb{C}^{N\times N}$, which is a rank-one semi-definite matrix, i.e.,
	 $\bm {W}_k \succeq 0$  and $\text{Rank}(\bm {W}_k)=1,k\in \mathcal{K}$. Then, we construct a block diagonal matrix $\bm {W}$, which is defined as
	 \begin{align}
     \bm{W}=\text{BlDiag}(\{\bm{W}_k\}_{k=1}^K)	
           =\begin{bmatrix}
			&\bm {W}_1  &\bm{0}_{N \times N}  &\cdots  &\bm{0}_{N \times N} \\
			&\bm{0}_{N \times N}  &\bm {W}_2  &\cdots &\bm{0}_{N \times N} \\
			&\vdots  &\vdots  &\ddots  &\vdots \\
			&\bm{0}_{N \times N}  &\bm{0}_{N \times N}  &\cdots  &\bm {W}_K \\
		\end{bmatrix},
	 \end{align}
	 where $\text{BlDiag}( \{\bm{X}_k\}_{k=1}^K  )$ denotes the block diagonal operation of $\{\bm{X}_k\}_{k=1}^K$.
	 It is straightforward to  find that $\bm {W}$ is also a semi-definite matrix, the rank of which is equal to $K$. Similarly, to make the constraints concise, the  following variables are first introduced: $\bm {d}_{p,k}=[\bm {0}_N^H,\ldots,\bm{h}_{p,k}^H,\ldots,\bm {0}_N^H]^H$, $ \overline{\bm {d}}_{q,k}=[\bm {0}_N^H,\ldots,\bm{h}_{s,q}^H,\ldots,\bm {0}_N^H]^H$, $\bm {r}_{p,k}=[\bm {0}_N^H,\ldots,\bm {g}_{p,k}^H \bm{\Theta}_r \bm F,\ldots,\bm {0}_N^H]^H$, and $ \overline{\bm {r}}_{q,k}=[\bm {0}_N^H,\ldots,\bm {g}_{s,q}^H \bm{\Theta}_t \bm F,\ldots,\bm {0}_N^H]^H$. Then, let $\bm {H}_{1,k}=(\bm {d}_{p,k}+\bm{r}_{p,k})(\bm {d}_{p,k}+\bm{r}_{p,k})^H$, $\bm {H}_{2,k}=(\bm {d}_{p,k}-\bm{r}_{p,k})(\bm {d}_{p,k}-\bm{r}_{p,k})^H$, $\overline{\bm {H}} _{1,q,k}=(\overline{\bm {d}}_{q,k}+\overline{\bm {r}}_{q,k})(\overline{\bm {d}}_{q,k}+\overline{\bm {r}}_{q,k})^H$, $\overline{\bm {H}} _{2,q,k}=(\overline{\bm {d}}_{q,k}-\overline{\bm {r}}_{q,k})(\overline{\bm {d}}_{q,k}-\overline{\bm {r}}_{q,k})^H$, {$\bm {D}_k=\text{BlDiag}(\{\bm h_{p,k} \bm h_{p,k}^H \}_{i=1}^K) - \bm h_{p,k} \bm h_{p,k}^H $,}
	$\overline{\bm {D}}_{q,k}=\sum_{i=1,i\neq k}^{K} \hat \mu_i(\overline{\bm {d}}_{q,i} \overline{\bm{d}}_{q,i}^H)$, 
	 {$\bm {R}_k=\text{BlDiag}(\{(\bm {g}_{p,k}^H \bm{\Theta}_r \bm F)^H (\bm {g}_{p,k}^H \bm{\Theta}_r \bm F)  \}_{i=1}^K) - \bm{r}_{p,k} \bm{r}_{p,k}^H$,}
	   $\overline{\bm {R}}_{q,k}=\sum_{i=1,i \neq k}^{K} (\overline{\bm {r}}_{q,i} \overline{\bm {r}}_{q,i}^H)$,  $\overline{\bm {R}}_q=\sum_{i=1}^{K} (\overline{\bm {r}}_{q,i} \overline{\bm {r}}_{q,i}^H)$, and $\overline{\bm {H}}_q=\sum_{i=1}^{K} \mu_i(\overline{\bm {d}}_{q,i} \overline{\bm {d}}_{q,i}^H)$. Note that $\text{BlDiag}( \{\bm{X}_k\}_{i=1}^K  )$ denotes a block diagonal matrix in which all the diagonal matrices are $\bm{X}_k$.
	   Based on the introduced variables above, the constraints C13-15 can be respectively reformulated as 
	\begin{align} 
		\label{broadcast_C1}
			&\frac{1}{2}\log_2(\text{Tr}(\bm {B}_{1,k} \bm W)+\sigma_{p,k}^2 )+\frac{1}{2}\log_2(\text{Tr}(\bm {B}_{2,k} \bm W)+\sigma_{p,k}^2 ) \nonumber \\
			&-\log_2(\text{Tr}(\bm {B}_{3,k} \bm W)+\sigma_{p,k}^2 ) \ge {R_{u,\text{min}}},~k \in \mathcal{K}, \\
			\label{broadcast_C2}
			&\frac{1}{2}\log_2(\text{Tr}(\overline{\bm {B}}_{1,q,k} \bm W)+\sigma_{s,q}^2 )+\frac{1}{2}\log_2(\text{Tr}(\overline{\bm {B}}_{2,q,k} \bm W)+\sigma_{s,q}^2 )\nonumber \\
			&-\log_2(\text{Tr}(\overline{\bm {B}}_{3,q,k} \bm W)+\sigma_{s,q}^2 ) \ge {R_{u,\text{min}}},~k\in \mathcal{K},~q\in\mathcal{Q}, \\
			&\overline{\text{C15}}: \text{Tr}((L\overline{\bm {R}}_q-{\gamma_{u,\text{min}}}\overline{\bm {H}}_q) \bm W)\ge \sigma_{s,q}^2 {\gamma_{u,\text{min}}},~q\in\mathcal{Q},		
	\end{align} 
	where $ \bm {B}_{1,k}=\bm {H}_{1,k}+\bm {B}_{3,k} $, $\bm {B}_{2,k}=\bm {H}_{2,k}+\bm {B}_{3,k}$, $\bm {B}_{3,k}=\bm {D}_k+\bm {R}_k$, $\overline{\bm {B}} _{1,q,k}=\overline{\bm {H}} _{1,q,k}+\overline{\bm {B}} _{3,q,k}$, $\overline{\bm {B}} _{2,q,k}=\overline{\bm {H}} _{2,q,k}+\overline{\bm {B}} _{3,q,k}$,  and $ \overline{\bm {B}} _{3,q,k}=\overline{\bm {D}} _{q,k}+\overline{\bm {R}} _{q,k} $. However, \textbf{P5} is still intractable after the above transformations due to  the difference of convex functions. To address this challenge, we can transform the constraints C13 and C14 into a convex form by exploiting their corresponding lower-bounds. Specifically, at any feasible point $\bm {W}^{(x)}$,  the lower-bounds  of  \eqref{broadcast_C1} and \eqref{broadcast_C2} by exploiting the first-order Taylor approximation are respectively given by
	\begin{align}
		\label{broadcast_C1_SCA}
		\overline {\text{C13} }:	
			&\frac{1}{2}\log_2(\text{Tr}(\bm {B}_{1,k} \bm W)+\sigma_{p,k}^2 )+\frac{1}{2}\log_2(\text{Tr}(\bm {B}_{2,k} \bm W)+\sigma_{p,k}^2 ) \nonumber \\
			&-\log_2(\text{Tr}(\bm {B}_{3,k} \bm W^{(x)})+\sigma_{p,k}^2 ) -\frac{\text{Tr}(\bm {B}_{3,k}(\bm W-\bm W^{(x)}))}{(\text{Tr}(\bm {B}_{3,k}\bm W^{(x)})+\sigma_{p,k}^2)\ln2} \ge {R_{u,\text{min}}}, \\
		\label{broadcast_C2_SCA}	
		\overline {\text{C14} }:	
			&\frac{1}{2}\log_2(\text{Tr}(\overline{\bm {B}}_{1,q,k} \bm W)+\sigma_{s,q}^2 )+\frac{1}{2}\log_2(\text{Tr}(\overline{\bm {B}}_{2,q,k} \bm W)+\sigma_{s,q}^2 ) \nonumber \\
			&-\log_2(\text{Tr}(\overline{\bm {B}}_{3,q,k} \bm W^{(x)})+\sigma_{s,q}^2 )  -\frac{\text{Tr}(\overline{\bm {B}}_{3,q,k}(\bm W-\bm W^{(x)}))}{(\text{Tr}(\overline{\bm {B}}_{3,q,k}\bm W^{(x)})+\sigma_{s,q}^2)\ln2} \ge {R_{u,\text{min}}}.	
	\end{align} 
	Then, by relaxing  the rank-$K$ constraint on $\bm W$ and using  $\overline {\text{C13} }$-$\overline {\text{C15} }$, \textbf{P5} is reformulated as  
	\begin{equation}\tag{$\textbf{P5.1}$} 
		\begin{aligned}
			\min_{W} ~~& \text{Tr}(\bm W) \\ 
			\text{s.t.}~~ &\overline{\text{C13} },~ \overline{\text{C14} },~ \overline{\text{C15}}.
		\end{aligned}
		\end{equation}
	It can be found that \textbf{P5.1} is convex, which  thus can be solved by implementing the  CVX tool.
  \begin{proposition} 
  \label{Rank-K}
  The optimal solution to \textbf{P5.1} in the $x$-th iteration, denoted by $\bm{W}^{(x+1)}$, is expressed as 
  \begin{align}
   \bm{W}^{(x+1)} = \text{BlDiag}\left(\{\bm{W}_k^*\}_{k=1}^K\right),	
 \end{align}
where  $\bm {W}_k^*$ is a rank-one matrix for any $k \in \mathcal{K}$, and  the rank of $\bm{W}^{(x+1)}$ is $K$.
 \end{proposition}

 \begin{proof}
Please refer to Appendix \ref{Proposition4}.
 \end{proof}

 Proposition \ref{Rank-K} reveals that relaxing the rank-$K$ constraint on $\bm{W}$ will not affect the  solution accuracy when solving \textbf{P5.1}. Then, we have $\bm {W}_k^* = \bm{w}_k^* {\bm{w}_k^*}^H$, where $\bm{w}_k^*$ is the optimal solution to \textbf{P5} in the $x$-th iteration. By iteratively solving \textbf{P5.1}  to update  $\bm{w}_k^*$, we can finally derive the near-optimal design of  active beamforming vectors at the BS when the convergence is achieved.

	\subsection{Reflection and Transmission Coefficients Optimization}
	With the active beamforming vectors obtained in Section \ref{subsectionA}, we proceed to find the   solution of the reflection and transmission coefficients by solving \textbf{P6}, which is formulated  as
		\begin{equation}\tag{$\textbf{P6}$} 
		\begin{aligned}
			&\text{Find} ~\{\bm{v}_r,\bm{v}_t\} \\ 
			& \text{s.t.}~~  \text{C4}-\text{C7},~\text{C13}-\text{C15}.
		\end{aligned}
	\end{equation}
It should be noted that \textbf{P6} can be solved by applying the method used in Section \ref{multi_sectionB}.
Similarly,  $ \tilde{\bm {v}}_r$ and $ \tilde{\bm {v}}_t$ are introduced to deal with the constraint C7, and \textbf{P6} is then recast as 
	\begin{equation}\tag{$\textbf{P6.1}$} 
		\begin{aligned}
			&\text{Find} ~\{\bm{v}_r,\bm{v}_t,\tilde{ \bm{v}}_r,\tilde{\bm{v}}_t\} \\ 
			 \text{s.t.}~~& \text{C4}-\text{C6},~\text{C9}-\text{C11},~\text{C13}-\text{C15}.
		\end{aligned}
	\end{equation}
	To solve \textbf{P6.1} efficiently, the PDD  and Lagrangian dual methods are also applied. Specifically, the solution of  $ \{\tilde{ \bm{v}}_r, \tilde{ \bm{v}}_t \}$ and $\{\bm{v}_r, \bm{v}_t \}$ in \textbf{P6.1} can be derived by solving the sub-problems defined in  \textbf{P3.5} and \textbf{P6.2}, respectively. As the solving process of $ \{\tilde{ \bm{v}}_r, \tilde{ \bm{v}}_t \}$  has been discussed in Section \ref{multi_sectionB}, we focus on solving \textbf{P6.2} formulated as follows
		\begin{equation}\tag{$\textbf{P6.2}$} 
			\begin{aligned}
				\min_{\bm{v}_r,\bm{v}_t, \tilde{\bm \alpha}} &-\sum_{k=1}^{K(Q+1)+Q}\tilde{\alpha} _k +\frac{1}{2\rho}\sum_{\iota \in\{r,t\}}\|\tilde{ \bm{v}}_{\iota}-\bm{v}_{\iota}+\rho\bm{\lambda_{\iota}}   \| ^2 \\ 
				\text{s.t.}~~& {{\widetilde{\text{C13}}}}: {R}_{u,k}^{(p)}  \ge R_{u,\text{min}}+\tilde {\alpha}_k, ~k\in \mathcal{K},\\
				& {\widetilde {\text{C14} }}:{R}_{q,k}^{(s)} \ge R_{u, \text{min}} + \tilde {\alpha}_{qK+k},~k\in \mathcal{K},~q\in\mathcal{Q}, \\
				& {\widetilde {\text{C15} }}:{\gamma}_{u,q} \ge \gamma_{u, \text{min}} + {\tilde {\alpha}_{K(Q+1)+q}},~q\in\mathcal{Q}, ~{\text{C12} },
			\end{aligned}
		\end{equation}	
	where $ \tilde {\bm {\alpha}}=[\tilde{\alpha}_1,\ldots,\tilde{\alpha}_{K(Q+1)+Q}]^T$ is the residual variable vector.  Let $\hat {d}_{k,i} = \bm{h}_{p,k}^H \bm w_i$, $\hat{\bm {r}}_{k,i}^H=\bm {g}_{p,k}^H\text{diag}(\bm F \bm w_i) $, $\tilde {d}_{q,k} = {\bm{h}_{s,q}^H}  \bm w_k$, $\tilde{\bm {r}}_{q,k}^H={\bm {g}_{s,q}^H}\text{diag}(\bm F \bm w_k) $, $\hat{c}_k=\text{Tr}(\bm {D}_k \bm W)$, $\tilde{c}_{q,k}=\text{Tr}(\overline{\bm {D}}_{q,k} \bm W)$, $\hat{\bm R}_k=\sum_{i=1,i\neq k}^{K}(\hat{\bm {r}}_{k,i} \hat{\bm {r}}_{k,i}^H)$, $\tilde{\bm R}_{q,k}=\sum_{i=1,i\neq k}^{K}(\tilde{\bm {r}}_{q,i} \tilde{\bm {r}}_{q,i}^H)$,  $\tilde{\bm R}_q=\sum_{i=1}^{K}(\tilde{\bm {r}}_{q,i} \tilde{\bm {r}}_{q,i}^H)$, and $\tilde \chi_{q}=\text{Tr}(\overline{\bm {H}}_{q} \bm W)$. Then, the left-hand sides of the constraints {${\widetilde {\text{C13} }}-{\widetilde {\text{C15} }}$} are respectively recast as
	\begin{align} 
	\label{uni_C1formulate}
	&R_{u,k} = \frac{1}{2}\log_2(1+{|\hat{a}_{1,k}|^2}/{\hat{b}_k})+\frac{1}{2}\log_2(1+{|\hat{a}_{2,k}|^2}/{\hat{b}_k}),~k\in \mathcal{K},  \\
	\label{uni_C2formulate}
	&R_{q,k} = \frac{1}{2}\log_2(1+{|\tilde{a}_{1,q,k}|^2}/{\tilde{b}_{q,k}		})+\frac{1}{2}\log_2(1+{|\tilde{a}_{2,q,k}|^2}/{\tilde{b}_{q,k}}),~k\in \mathcal{K},~q\in \mathcal{Q}, \\  
	\label{uni_C3formulate}
	& \gamma_{u,q} = {L} (\bm {v}_t^H\tilde{\bm R}_q \bm {v}_t)/{(\tilde \chi_{q}+\sigma_{s,q}^2)},~q\in \mathcal{Q},		
	\end{align}	
	where $\hat{a}_{1,k}=\hat {d}_{k,k}+\hat{\bm {r}}_{k,k}^H \bm {v}_r$, $\hat{a}_{2,k}=\hat {d}_{k,k}-\hat{\bm {r}}_{k,k}^H \bm {v}_r$, $\hat{b}_k=\hat{c}_k+\bm {v}_r^H\hat{\bm R}_k\bm {v}_r+\sigma_{p,k}^2$,  $\tilde{a}_{1,q,k}=\tilde {d}_{q,k}+\tilde{\bm {r}}_{q,k}^H \bm {v}_t$, $\tilde{a}_{2,q,k}=\tilde {d}_{q,k}-\tilde{\bm {r}}_{q,k}^H \bm {v}_t$, and $\tilde{b}_{q,k}=\tilde{c}_{q,k}+\bm {v}_t^H\tilde{\bm R}_{q,k}\bm {v}_t+\sigma_{s,q}^2$. Since \eqref{uni_C1formulate}, \eqref{uni_C2formulate} and \eqref{uni_C3formulate} are all non-convex, which results in the challenge of solving \textbf{P6.2}. To transform \textbf{P6.2} into a convex one, we also exploit the lower-bounds of \eqref{uni_C1formulate}, \eqref{uni_C2formulate} and \eqref{uni_C3formulate}, which are respectively given by
		\begin{align} 
			\label{uni_C1formulate_lower}
			&R_{u,k}^{(p)} \ge \frac{1}{2 \ln2}( \ln(1+ \frac{|\hat{a}_{1,k}^{(x)}|^2}{\hat{b}_k^{(x)}} ) - \frac{|\hat{a}_{1,k}^{(x)}|^2}{\hat{b}_k^{(x)}} +\frac{2\mathcal{R}\{(\hat{a}_{1,k}^{(x)})^H\hat{a}_{1,k}\} }{\hat{b}_k^{(x)}} - \frac{|\hat{a}_{1,k}^{(x)}|^2(\hat{b}_k+|\hat{a}_{1,k}|^2)}{\hat{b}_k^{(x)}(\hat{b}_k^{(x)}+|\hat{a}_{1,k}^{(x)}|^2)}   +   \ln(1+ \frac{|\hat{a}_{2,k}^{(x)}|^2}{\hat{b}_k^{(x)}} ) \nonumber \\
			&~~~~ - \frac{|\hat{a}_{2,k}^{(x)}|^2}{\hat{b}_k^{(x)}}           
			+ \frac{2\mathcal{R}\{(\hat{a}_{2,k}^{(x)})^H\hat{a}_{2,k}\} }{\hat{b}_k^{(x)}} -\frac{|\hat{a}_{2,k}^{(x)}|^2(\hat{b}_k+|\hat{a}_{2,k}|^2)}{\hat{b}_k^{(x)}(\hat{b}_k^{(x)}+|\hat{a}_{2,k}^{(x)}|^2)} )  \triangleq \widetilde {R}_{u,k}^{(p)},~k \in \mathcal{K},\\
			\label{uni_C2formulate_lower}
			& R_{q,k}^{(s)} \ge \frac{1}{2 \ln2}( \ln(1+ \frac{|\tilde{a}_{1,q,k}^{(x)}|^2}{\tilde{b}_{q,k}^{(x)}} ) - \frac{|\tilde{a}_{1,q,k}^{(x)}|^2}{\tilde{b}_{q,k}^{(x)}} +\frac{2\mathcal{R}\{(\tilde{a}_{1,q,k}^{(x)})^H\tilde{a}_{1,q,k}\} }{\tilde{b}_{q,k}^{(x)}} \nonumber \\
			&~~~ - \frac{|\tilde{a}_{1,q,k}^{(x)}|^2(\tilde{b}_{q,k}+|\tilde{a}_{1,q,k}|^2)}{\tilde{b}_{q,k}^{(x)}(\tilde{b}_{q,k}^{(x)}+|\tilde{a}_{1,q,k}^{(x)}|^2)}   +   \ln(1+ \frac{|\tilde{a}_{2,q,k}^{(x)}|^2}{\tilde{b}_{q,k}^{(x)}} ) - \frac{|\tilde{a}_{2,q,k}^{(x)}|^2}{\tilde{b}_{q,k}^{(x)}}   + \frac{2\mathcal{R}\{(\tilde{a}_{2,q,k}^{(x)})^H\tilde{a}_{2,q,k}\} }{\tilde{b}_{q,k}^{(x)}} \\
			&~~~-\frac{|\tilde{a}_{2,q,k}^{(x)}|^2(\tilde{b}_{q,k}+|\tilde{a}_{2,q,k}|^2)}{\tilde{b}_{q,k}^{(x)}(\tilde{b}_{q,k}^{(x)}+|\tilde{a}_{2,q,k}^{(x)}|^2)} )  \triangleq \widetilde {R}_{q,k}^{(s)},~k \in \mathcal{K},~q\in \mathcal{Q}, \nonumber\\
			\label{uni_C3formulate_lower}
			& \gamma_{u,q}  \ge {\frac{L}{\tilde \chi_{q}+\sigma_{s,q}^2}}( (\bm {v}_t^{(x)})^H\tilde{\bm R}_q {v}_t^{(x)}  +  \mathcal{R}\{(\bm {v}_t^{(x)})^H (\tilde{\bm R}_q+\tilde{\bm R}_q^H)(\bm {v}_t-{v}_t^{(x)})  \}) \triangleq \widetilde{\gamma}_{u,q},~q\in \mathcal{Q}.	
		\end{align}		
Based on \eqref{uni_C1formulate_lower}, \eqref{uni_C2formulate_lower} and \eqref{uni_C3formulate_lower}, \textbf{P6.2} can be rewritten as
 \begin{equation}\tag{$\textbf{P6.3}$} 
 	\begin{aligned}
 		\min_{\bm{v}_r,\bm{v}_t, \tilde{\bm \alpha}} & -\sum_{k=1}^{K(Q+1)+Q}\tilde{\alpha} _k +\frac{1}{2\rho}\sum_{\iota \in\{r,t\}}\|\tilde{ \bm{v}}_{\iota}-\bm{v}_{\iota}+\rho\bm{\lambda_{\iota}}   \| ^2 \\ 
 		\text{s.t.}~~& \widetilde{{\widetilde{\text{C13}}}}:{\widetilde {R}_{u,k}^{(p)}}  \ge R_{u,\text{min}} +\tilde {\alpha}_k, ~k\in \mathcal{K},\\
 		&\widetilde {\widetilde {\text{C14} }}:\widetilde{R}_{q,k}^{(s)} \ge {R_{u,\text{min}}}+\tilde {\alpha}_{qK+k},~k\in \mathcal{K},~q\in \mathcal{Q}, \\
 		&  \widetilde{\widetilde {\text{C15} }}:\widetilde{\gamma}_{u,q} \ge {\gamma_{u,\text{min}}}+{\tilde {\alpha}_{K(Q+1)+q}},~q\in \mathcal{Q},~  {\text{C12} }.
 	\end{aligned}
 \end{equation}
Similar to \textbf{P3.4}, \textbf{P6.1} is a convex problem and  can be solved by implementing the CVX tool.  
 
	The algorithm for solving \textbf{P4}  is similar to Algorithm \ref{AlgorithmA} and omitted here for simplicity.  The overall complexity for solving \textbf{P4} is  $\mathcal{O}(I_1I_2I_4\sqrt{K(Q+1+N)+Q} \log(\frac{1}{\epsilon})\cdot(n_1(K(Q+1)+Q+(KN)^3)+ 
		n_1^2(K(Q+1)+Q+(KN)^2)+n_1^3  )   + I_1 I_2 I_3\sqrt{2(K(Q+1)+Q+M)} \log(\frac{1}{\epsilon})\cdot(n_2( 2(K(Q+1)+Q)) 
		+n_2^2(2(K(Q+1)+Q))+n_2(4M)+n_2^3  ) )$, where $n_1=(KN)^2$, $n_2=2M+K+KQ+Q$, and $I_4$ denotes the iteration times for solving  \textbf{P5.1} by applying the  SCA technique.

	\section{Numerical Results}
	\label{section5}
	\subsection{Simulation Setup}
	In this section, we conduct numerical simulations to verify the performance of the proposed schemes. We consider a two-dimensional coordinate model, in which the BS and STAR-RIS are located at $(0, 0)$ and (100 m, 0), respectively. The PUs are randomly distributed in a circle with a center (100 m, 20m) and a radius of 5 m, while the SUs are randomly distributed in a circle with a center (100 m, -20m) and a radius of 5 m. We consider that all signal transmissions suffer from the large-scale fading and small-scale fading. According to \cite{Zeng yong}, the large-scale fading for all links is modeled as $D(d) = (\zeta/(4\pi))^2 d^{-\varsigma}$, which is a function of the distance between  two nodes (i.e., $d$), the wavelength (i.e., $\zeta$), and the path-loss exponent (i.e., $ \varsigma $). The small-scale fading for the links related with the STAR-RIS is modeled as the Rician fading. For example, $\bm{F}$ is expressed as
	$$\bm{F} =  \sqrt{ \frac{ D(d_{BS,RIS}) \kappa_F}{\kappa_F + 1 }} \bm{F}^{\text{LoS}} + \sqrt{  \frac{D(d_{BS,RIS})  }{\kappa_F + 1}} \bm{F}^{\text{NLoS}},$$ 
	where $ \bm{F}^{\text{LoS}} $, $\bm{F}^{\text{NLoS}}$,  $\kappa_F$, and $d_{BS,RIS}$ are the line-of-sight (LoS) component,  non-LoS component and Rician factor, and the distance between the BS and the STAR-RIS, respectively.
	$\bm{F}^{\text{LoS}}$ is expressed as $\bm{F}^{\text{LoS}}= \bm{\beta}_M(\upsilon_\text{AoA}) \bm{\beta}_N^H(\upsilon_\text{AoD})$, where 
	$\bm{\beta}_{X}(\upsilon) = [1, e^{j\pi \text{sin}(\upsilon) }, \ldots, e^{j\pi (X-1) \text{sin}(\upsilon) }]^T$, $X\in \{M,N\}$,
	$\upsilon_\text{AoA}$ and $\upsilon_\text{AoD}$ are the angle of arrival and angle of departure, respectively.
	$\bm{F}^{\text{NLoS}}$ is modeled as the Rayleigh fading, each element of which satisfies  $\mathcal{CN}(0,1)$. Similarly, the definitions of $\bm{g}_{p,k}^H$  and $\bm{g}_{s,q}^H$ can be obtained.
	However, the small-scale fading for the links unrelated with the STAR-RIS  is modeled as  the Rayleigh fading with  $\mathcal{CN}(0,1)$. The path-loss exponents from the BS to PUs/SUs, from the STAR-RIS to PUs/SUs, and from the BS to STAR-RIS are set at 3.8, 2, and 2.4, respectively.
	The other parameters are set as follows: $N=4$, $M=20$, $K=4$, $Q=1$, $L=50$, $R_{b,\text{min}} = 2 $bit/s/Hz, $R_{u,\text{min}}=0.42$ bit/s/Hz, $\gamma_{b, \text{min}} = \gamma_{u,\text{min}} = 30$ dB, $\sigma_{p,k}^2=\sigma_{s,q}^2=-90$ dBm,  $\zeta = 750$ MHz, and $\mu=\mu_k=0.01$.

	Three baseline schemes are considered for verifying the superiority of the proposed schemes.
	\begin{itemize}
		\item{Baseline scheme 1: The amplitudes of  reflection and transmission coefficients at each element  are equal, i.e., $\beta_{m}^r=\beta_{m}^t=0.5$, while the other variables are jointly optimized without the phase correlation constraint.}
		\item{Baseline scheme 2: The amplitudes of the reflection  and transmission coefficients at each element  are equal, i.e., $\beta_{m}^r=\beta_{m}^t=0.5$, while the phase shifts are randomly generated without the phase correlation constraint.}	
		\item{Baseline scheme 3: The backscattering device with one antenna is considered as the secondary transmitter to replace the STAR-RIS, while the other variables are jointly optimized.}	
	\end{itemize}
It should be noted that  the baseline schemes are with the imperfect SIC in Figs. \ref{fig:mul_gammac}-\ref{fig:numberM_bro} but are with the perfect SIC in Figs. \ref{fig:numberN_bro}-\ref{fig:numberuser_uni}.

	Moreover, the proposed schemes without the phase correlation constraint are also investigated for evaluating the performance loss caused by this constraint.

	\subsection{Performance Evaluation}
	
\begin{figure}
	\centering
	\includegraphics[width=0.45\linewidth]{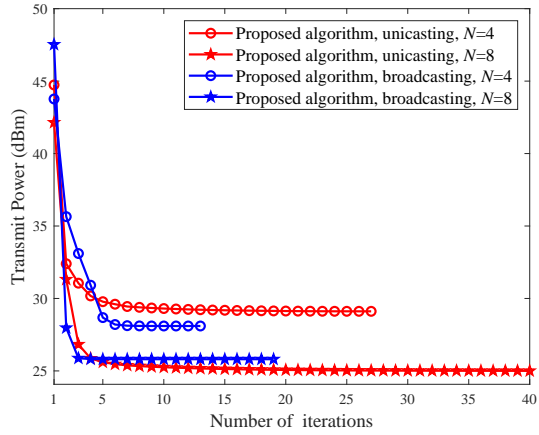}
	\caption{Transmit power versus the iteration process.}
	\label{fig:convengence}
\end{figure}

	The convergence performance of the proposed algorithms is investigated in Fig. \ref{fig:convengence}.  As seen from Fig. \ref{fig:convengence}, the value of  transmit power converges to a stationary point after several iterations, which confirms the convergence efficiency of the proposed algorithm for each signal model. 
	Specifically, it can be observed that the proposed algorithm for the unicasting signal model converges slower than that for the broadcasting signal model. It is because  more iterations are required to meet a larger number of the QoS constraints.  Moreover, the number of antennas at the BS is also a factor affecting the required iterations to converge, i.e., exploiting more antennas at the BS increases the required number of iterations.

	  \begin{figure*}
   \begin{minipage}[t]{0.50\linewidth}
    \centering
    \includegraphics[width=0.8 \linewidth]{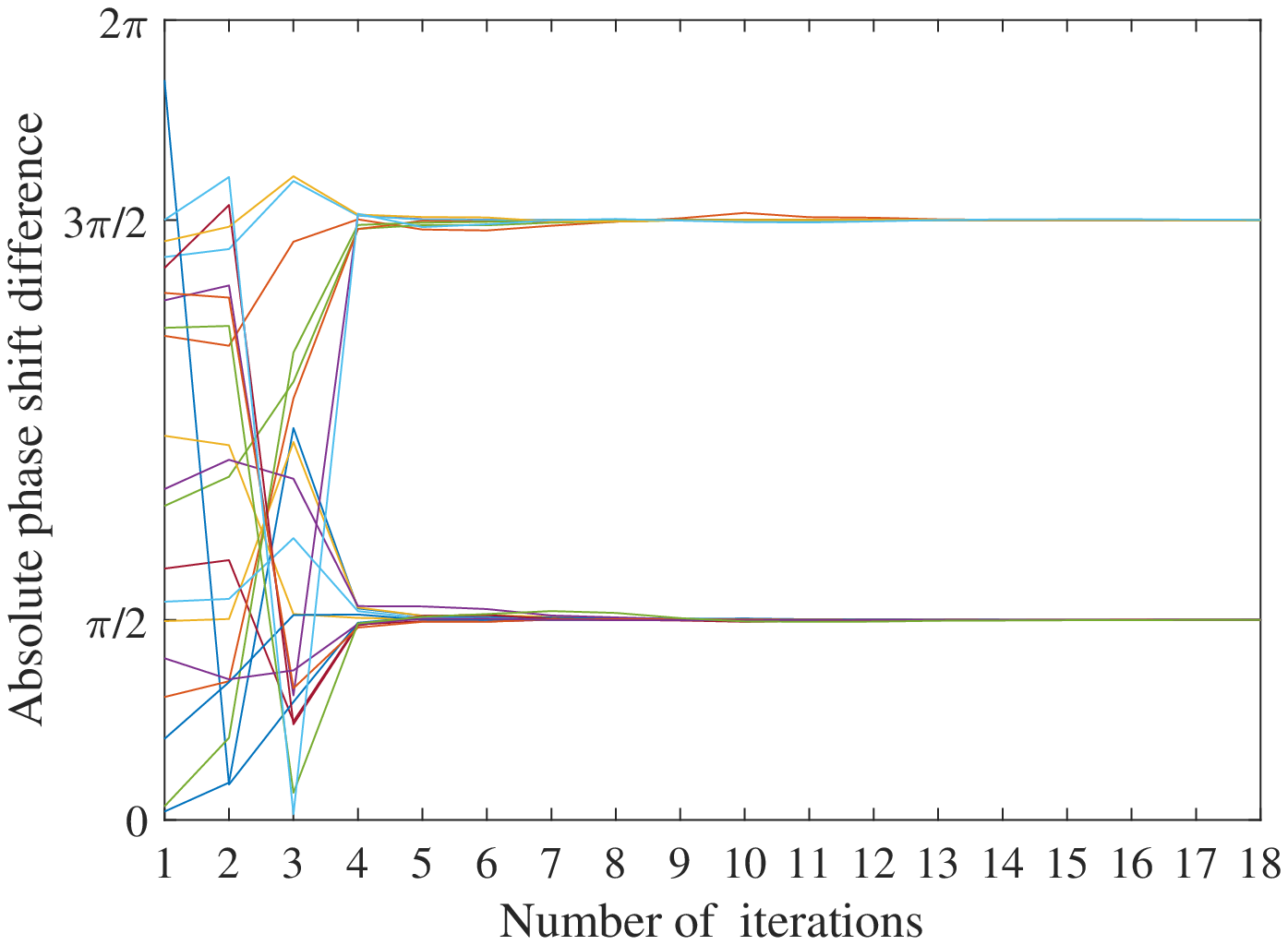}
    \caption{Absolute phase shift differences versus the iteration  process, where the lines of different colors represent the absolute differences between reflection phase shifts and  transmission phase shifts for STAR-RIS elements.}
    \label{fig:phasechange}
  \end{minipage}%
  \begin{minipage}[t]{0.50\linewidth}
    \centering
    \includegraphics[width=0.8 \linewidth]{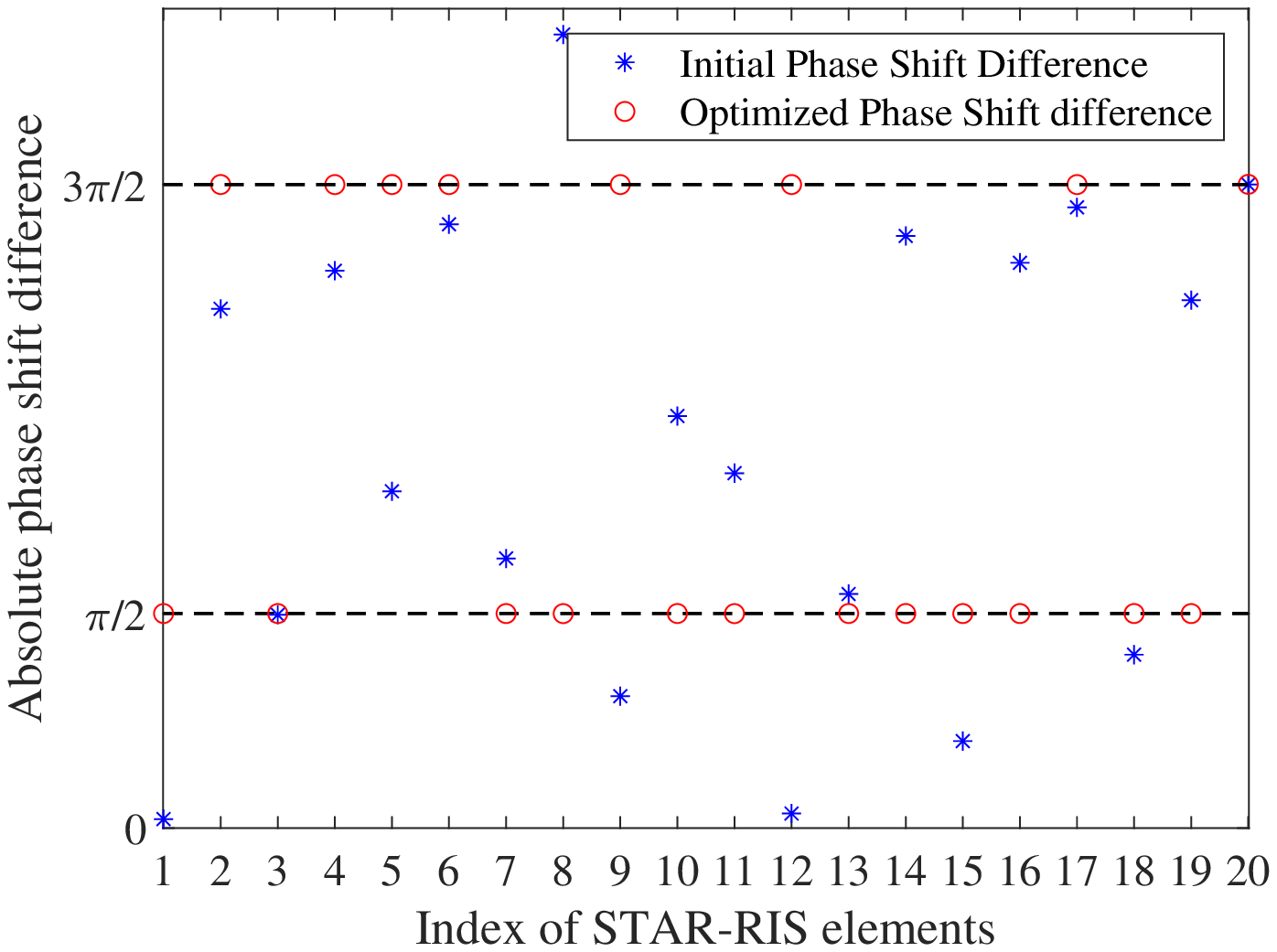}
    \caption{Scatter of absolute phase shift differences for the STAR-RIS.}
    \label{fig:scatterdif}
  \end{minipage}
\end{figure*}



	In Fig. \ref{fig:phasechange}, we depict the updating process of the absolute phase shift differences for all elements (i.e., $|\theta_m^r-\theta_m^t|$ for $m \in \mathcal{M}$). It can be observed that the absolute  phase shift differences for all elements are  irregular during the first four iterations, but then converge to $\frac{\pi}{2}$ or $\frac{3\pi}{2}$, which confirms that the phase correction constraint is satisfied.   
	Fig. \ref{fig:scatterdif} shows the effectiveness of phase shift optimization. In the initial step of the proposed algorithm, the phase shift differences are randomly distributed between $0$ and $2\pi$. After implementing the proposed algorithm, the phase shift  difference for each element will be $\frac{\pi}{2}$ or $\frac{3\pi}{2}$. This observation indicates that the phase correlation constraint affects the design of phase shifts for all elements.


	  \begin{figure*}
  \begin{minipage}[t]{0.50\textwidth}
    \centering
    \includegraphics[width=0.8 \linewidth]{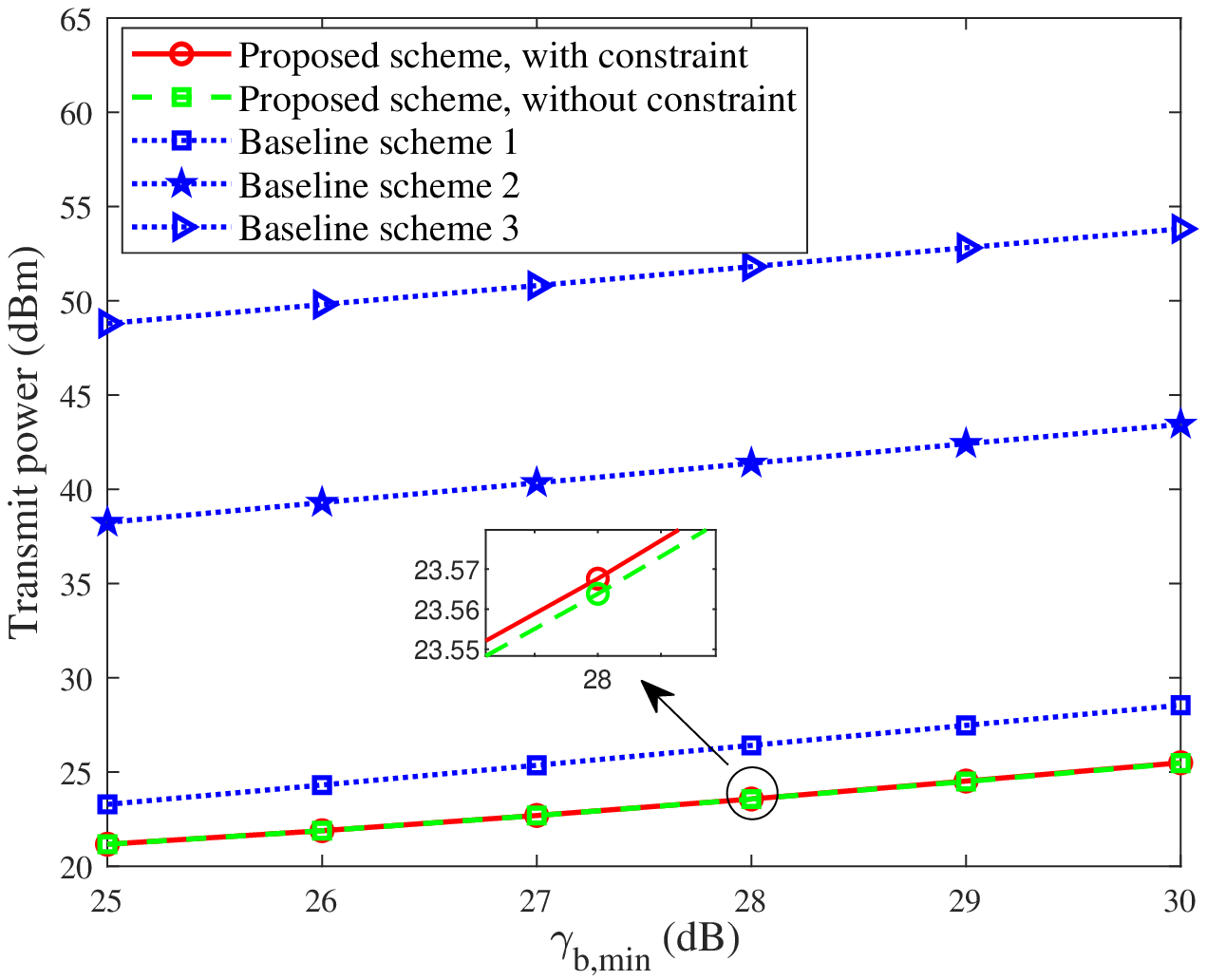}
    \caption{Transmit power versus the minimum SNR requirement \protect\\ for the broadcasting signal model.}
    \label{fig:mul_gammac}
  \end{minipage}%
  \begin{minipage}[t]{0.50\linewidth}
    \centering
    \includegraphics[width=0.8 \linewidth]{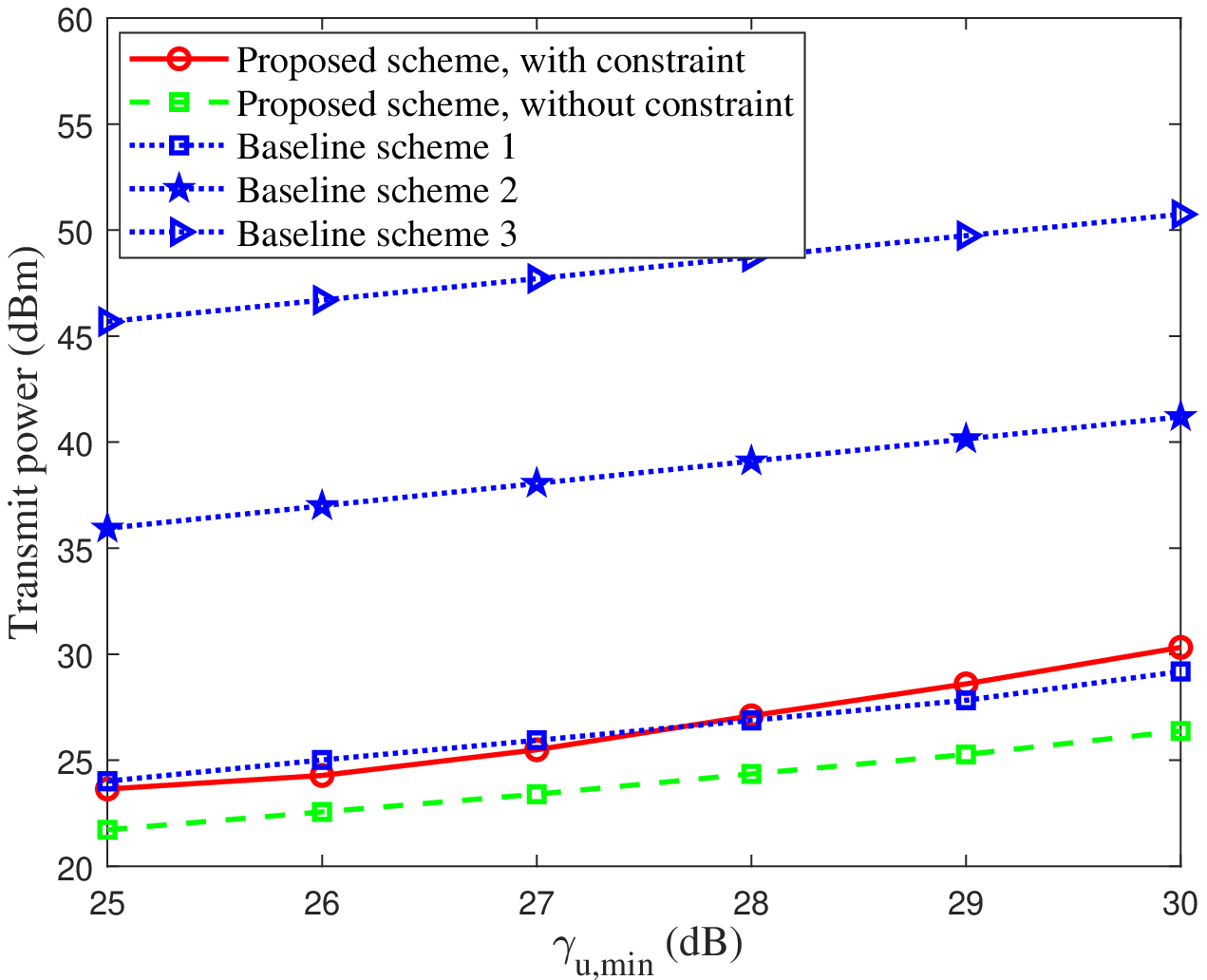}
    \caption{Transmit power versus the minimum SNR requirement \protect\\ for the unicasting signal model.}
    \label{fig:bro_gammac}
  \end{minipage}
\end{figure*}

The transmit power versus the minimum SNR requirement for the broadcasting signal model  is shown in Fig. \ref{fig:mul_gammac}. As observed in Fig. \ref{fig:mul_gammac}, a higher SNR requirement for decoding $c$ improves the transmit power. The reason is that more transmit power should be consumed for satisfying the SNR requirement. Compared to the baseline schemes, our proposed scheme consumes the lowest transmit power, which indicates the effectiveness of the proposed scheme. Specifically, from the comparison with the baseline scheme 1 and baseline scheme 2,  we observe that the optimization of  STAR-RIS coefficients (i.e., amplitude coefficients and phase shifts) is vital for enhancing the efficiency for simultaneous reflection and transmission empowered by the STAR-RIS. Moreover, the transmit power of the baseline scheme 3 is about twice as much as that of our proposed scheme, which confirms that the application of STAR-RIS can significantly improve the communication efficiency for the secondary transmission.
 Fig. \ref{fig:bro_gammac} investigates the impact of minimum SNR requirement on the transmit power for the unicasting signal model.  Similar to Fig. \ref{fig:mul_gammac},  our proposed scheme consumes less transmit power than the baseline scheme 2 and the baseline scheme 3.  Moreover, compared to the proposed scheme without the phase correlation constraint, the existence of this constraint degrades the system performance since the feasibility of designing phase shifts is constrained due to the fact that the phase shift difference for each element has to be $\frac{\pi}{2}$ or $\frac{3\pi}{2}$.



  \begin{figure*}
  \begin{minipage}[t]{0.50\textwidth}
    \centering
    \includegraphics[width=0.8 \linewidth]{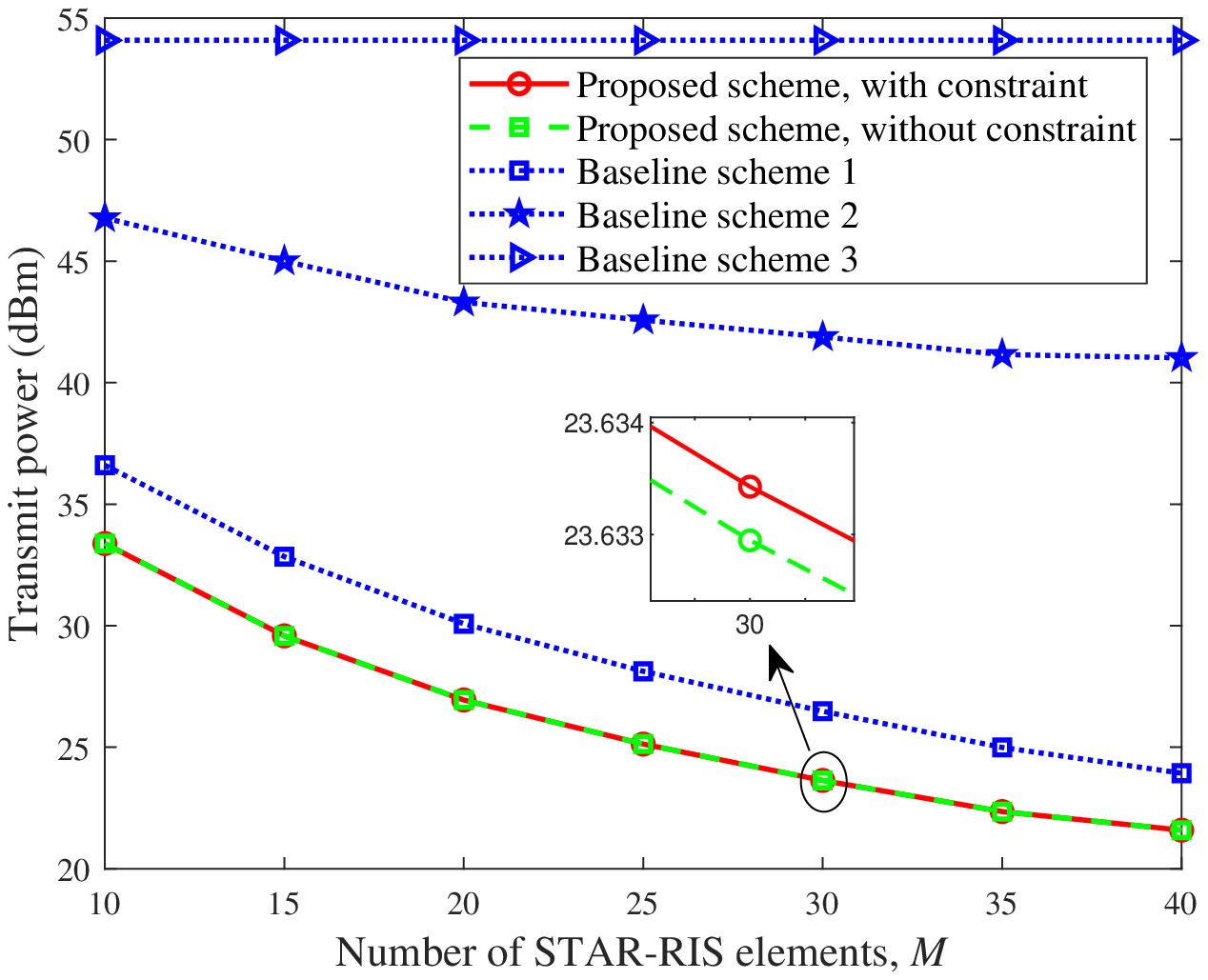}
    \caption{Transmit power versus the number of STAR-RIS elements \protect\\ for the broadcasting signal model}
    \label{fig:numberM_mul}
  \end{minipage}%
  \begin{minipage}[t]{0.50\linewidth}
    \centering
    \includegraphics[width=0.8 \linewidth]{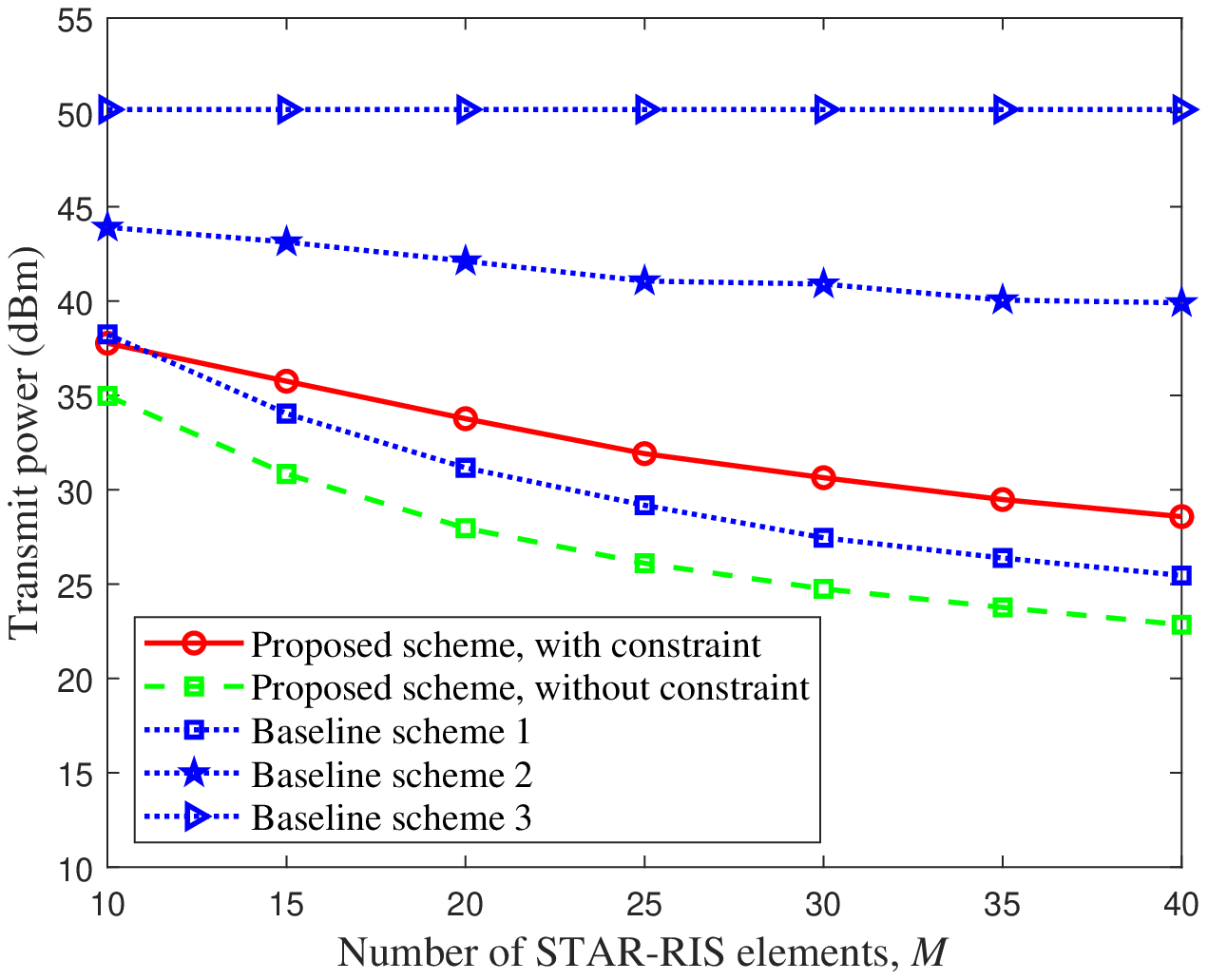}
    \caption{Transmit power versus the number of STAR-RIS elements \protect\\ for the unicasting signal model.}
    \label{fig:numberM_bro}
  \end{minipage}
\end{figure*}

We investigate the transmit power versus the number of STAR-RIS elements for both  models in Fig. \ref{fig:numberM_mul} and Fig. \ref{fig:numberM_bro}, respectively. As shown in Fig. \ref{fig:numberM_mul} and Fig. \ref{fig:numberM_bro}, by increasing the number of elements from $10$ to $40$, the required transmit power for satisfying all  constraints is reduced. It demonstrates that  deploying more  elements at the  STAR-RIS is an efficient way for enhancing system communication efficiency as more additional reflection and transmission links can be achieved for information delivery.
 Furthermore, it is seen that the gap between  the proposed scheme and the baseline schemes  is a non-decreasing function  with the number of STAR-RIS  elements.  Compared to the proposed scheme without the phase correlation constraint, the proposed scheme with the phase correlation constraint consumes more transmit power, especially for the unicasting signal model. The reason is that for the proposed scheme with the phase correlation constraint, more restrictive constraints are required. For the unicasting signal model, the proposed scheme with the phase correlation constraint even consumes a bit more power than the baseline scheme 1 when the number of elements is small, which again indicates the effect of the phase correction constraint on performance degradation. However, by removing this constraint, the proposed scheme can reduce up to 3.26 dBm transmit power compared to the baseline scheme 1.

 \begin{figure*}
	\begin{minipage}[t]{0.50\textwidth}
		\centering
		\includegraphics[width=0.8 \linewidth]{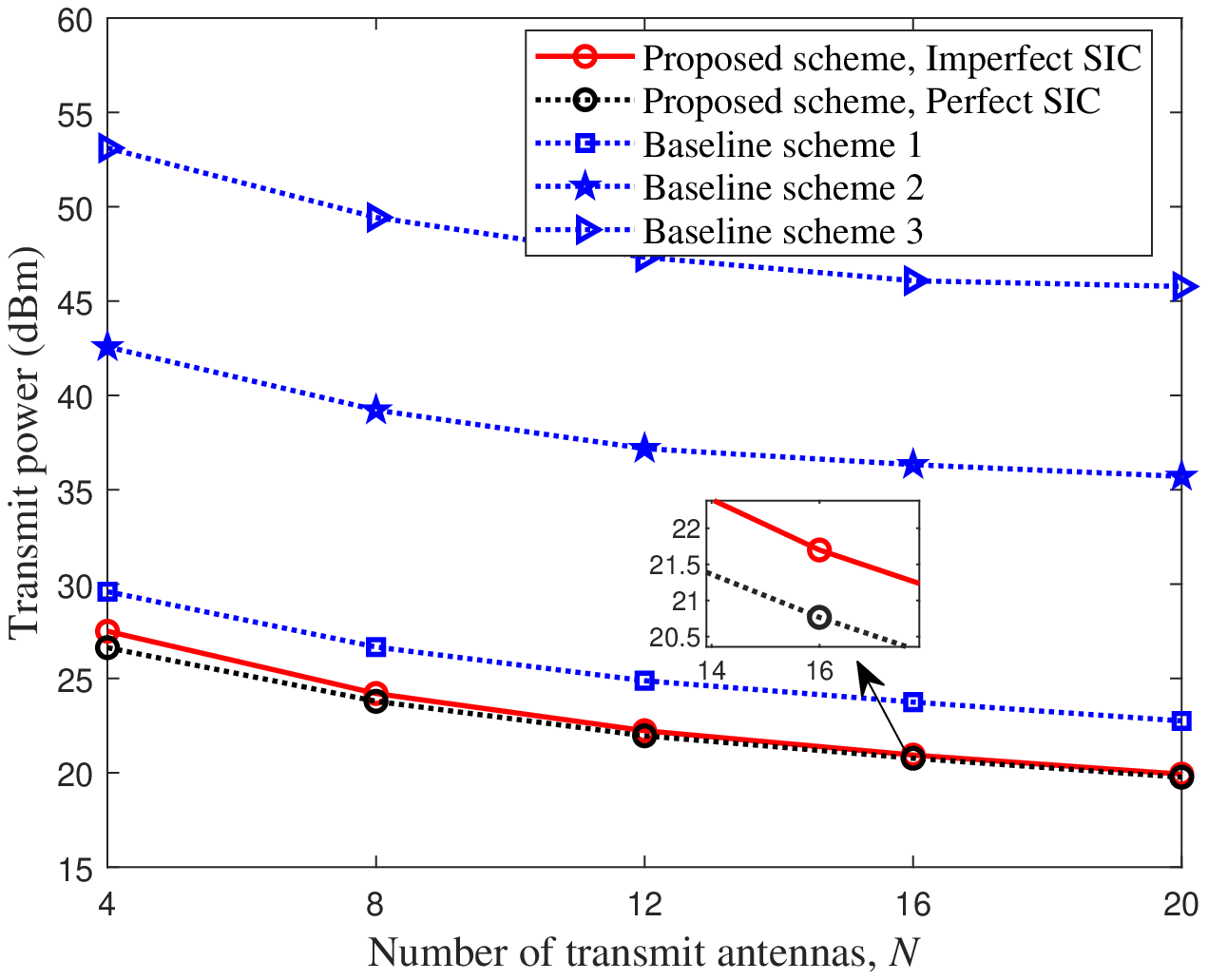}
		\caption{Transmit power versus the number of transmit antennas \protect\\ for the broadcasting signal model}
		\label{fig:numberN_bro}
	\end{minipage}%
	\begin{minipage}[t]{0.50\linewidth}
		\centering
		\includegraphics[width=0.8 \linewidth]{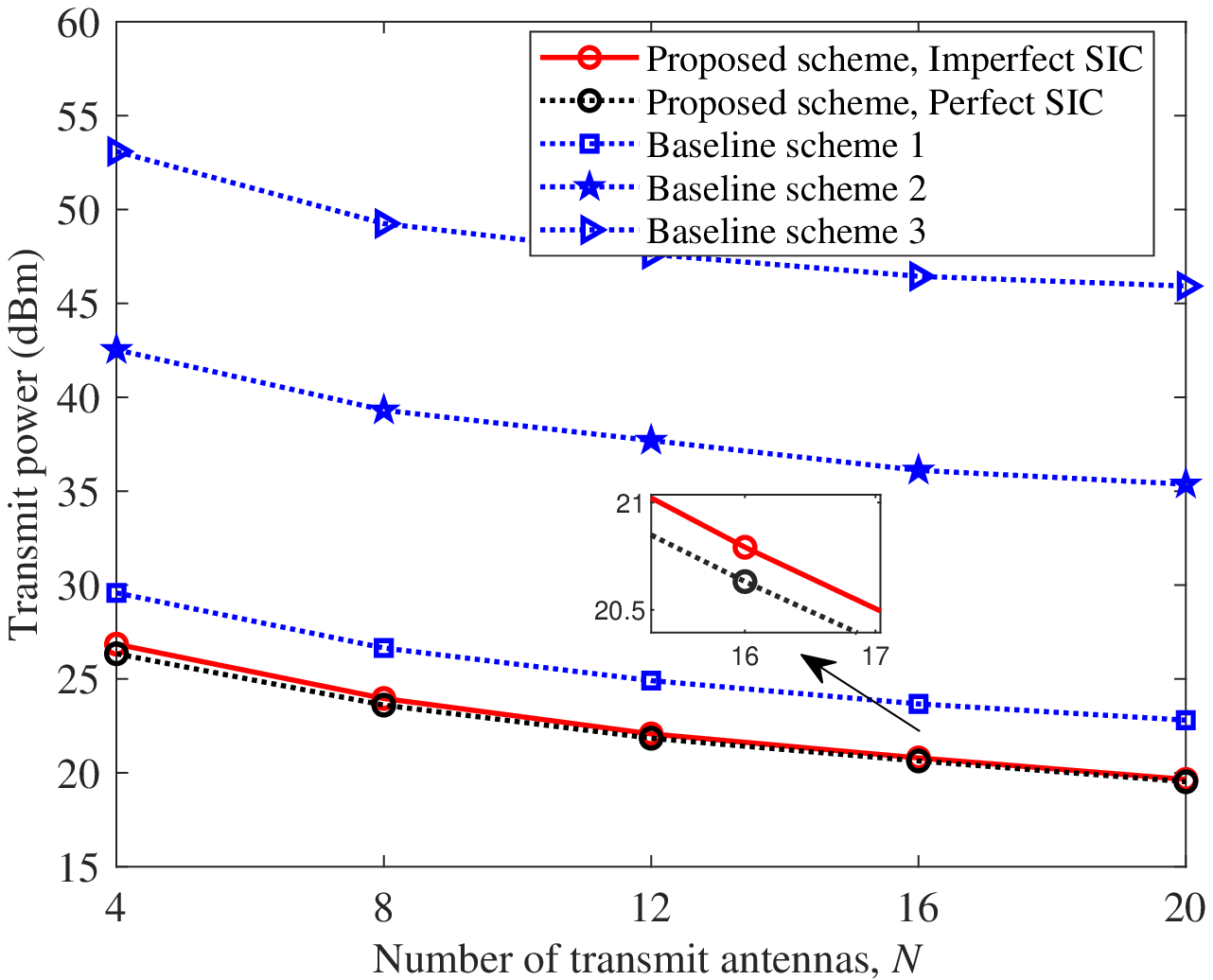}
		\caption{Transmit power versus the number of transmit antennas \protect\\ for the unicasting signal model.}
		\label{fig:numberN_uni}
	\end{minipage}
\end{figure*}
	
In Fig. \ref{fig:numberN_bro} and Fig. \ref{fig:numberN_uni}, the transmit power versus the number of transmit antennas for both models is investigated, where $K=2$ and  $\mu=\mu_k=0.04$.  As observed in Fig. \ref{fig:numberN_bro} and Fig. \ref{fig:numberN_uni}, increasing the number of antennas at the BS reduces
 the  required transmit power of all schemes. It is due to the fact that a higher antenna gain can be obtained by increasing the number of transmit antennas, which enhances the efficiency of both primary and secondary transmissions. In addition, compared to the baseline schemes with the perfect SIC, the proposed scheme with the imperfect SIC even consumes less transmit power, which verifies the effectiveness of the proposed scheme.

  \begin{figure*}
	\begin{minipage}[t]{0.50\textwidth}
		\centering
		\includegraphics[width=0.8 \linewidth]{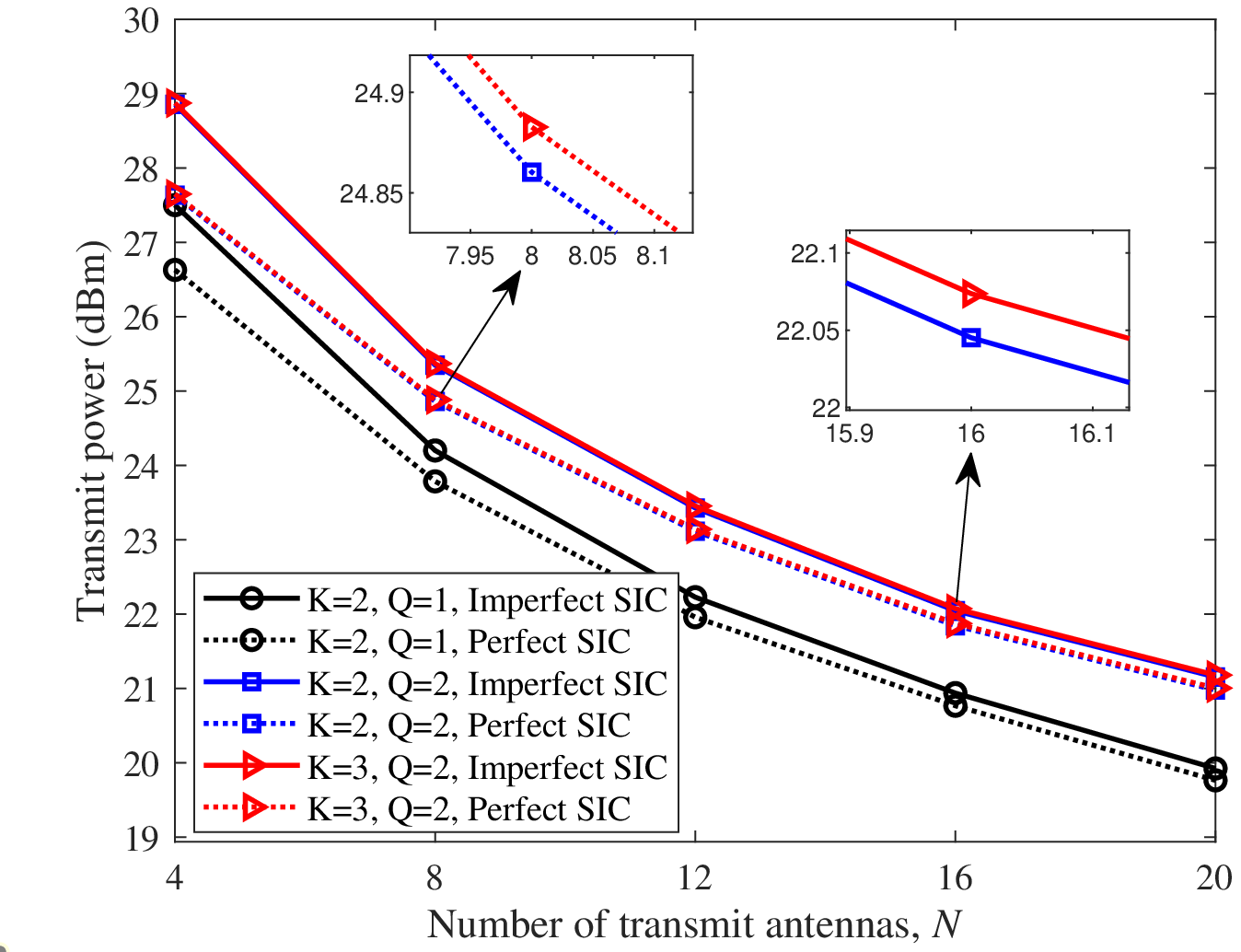}
		\caption{Transmit power versus the number of transmit antennas \protect\\ for the broadcasting signal model}
		\label{fig:numberuser_bro}
	\end{minipage}%
	\begin{minipage}[t]{0.50\linewidth}
		\centering
		\includegraphics[width=0.8 \linewidth]{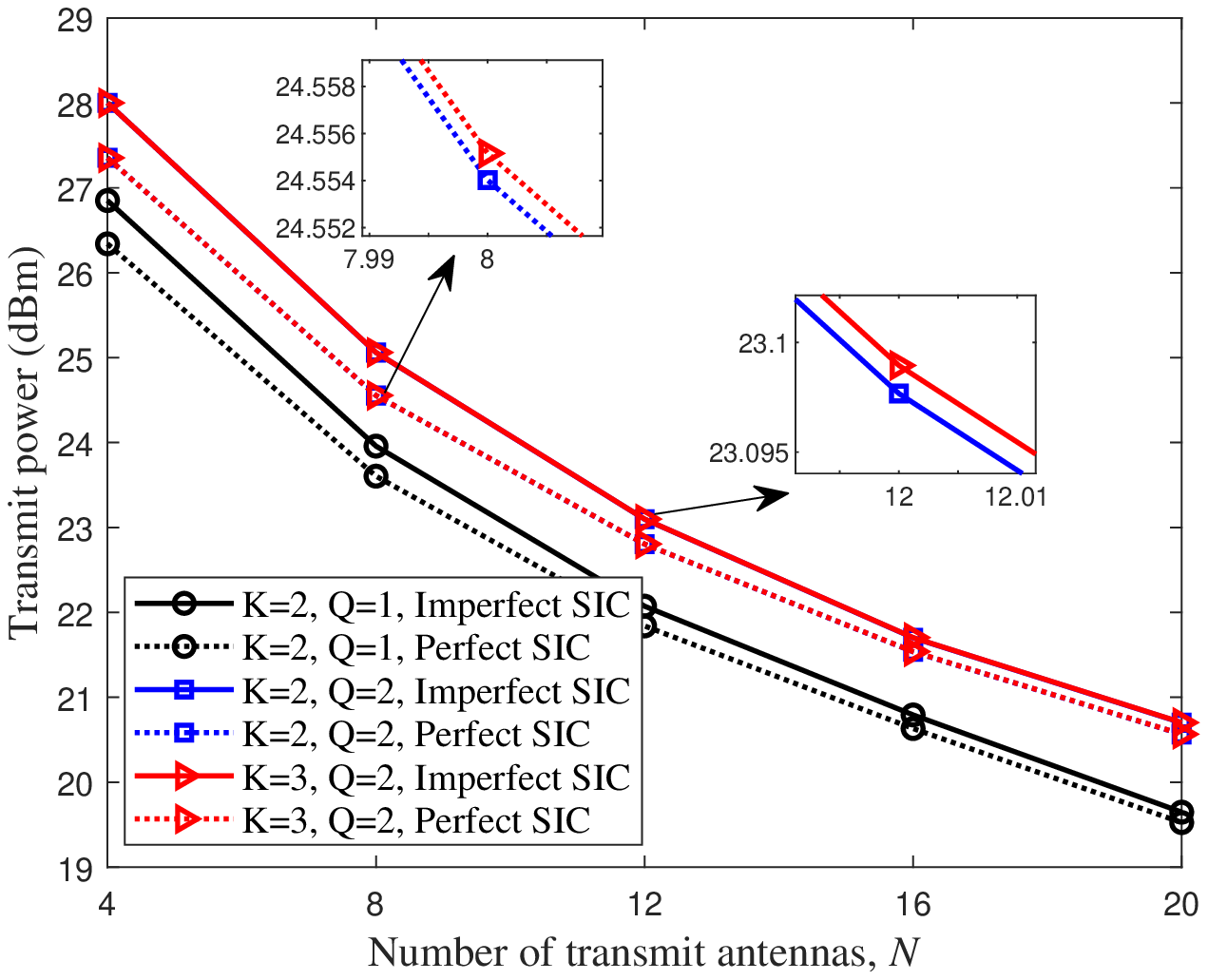}
		\caption{Transmit power versus the number of transmit antennas \protect\\ for the unicasting signal model.}
		\label{fig:numberuser_uni}
	\end{minipage}
\end{figure*}
	
Finally, we show the impact of the imperfect SIC on the transmit power for both models in Fig. \ref{fig:numberuser_bro} and Fig. \ref{fig:numberuser_uni}, respectively, where $\mu=\mu_k=0.04$. It can be found  that compared to the perfect SIC scenario, the imperfect SIC will result in more transmit power consumption at the BS.  The reason is that  the residual interference caused by the imperfect SIC reduces the SINR at the PUs and SUs, thus consuming a bit more transmit power to satisfy the QoS constraints. Moreover, the results in  Fig. \ref{fig:numberuser_bro} and Fig. \ref{fig:numberuser_uni} indicate that the number of PUs and SUs is an important factor affecting the transmit power at the BS. Specifically, the deployment of  more PUs and/or SUs may increase the transmit power. It is because a bit more transmit power is required to guarantee the QoS constraints corresponding to the worst-case of channel conditions.

	\section{Conclusions}
	\label{section6}
	In this paper, we have proposed a new idea of using   the STAR-RIS for SR systems to achieve three objectives, i.e., (i) guaranteeing that the PUs and SUs can be deployed on different sides of the STAR-RIS and thus providing a promising solution to support indoor and outdoor transmissions, (ii) enhancing the primary transmission efficiency  by constructing passive beamforming, and (iii) serving as a transmitter to passively deliver secondary information to the SUs. We have investigated the transmit power minimization problems for both broadcasting and unicasting signal models and proposed a BCD based algorithm with SDR, PDD, and SCA techniques to solve them. Numerical results have been conducted to demonstrate the superior performance of our proposed scheme. The numerical results reveal three important observations: 1) the phase correlation constraint results in  the increase of transmit power due to the reduced feasibility for designing phase shifts; 2) the optimization of simultaneous reflection and transmission coefficients is an important factor of enhancing system performance; 3) the STAR-RIS without optimizing reflection and transmission coefficients can even bring a higher system efficiency compared to the backscattering device enabled scheme.

	\appendices

		\section{Proof of Proposition \ref{Rank-K}}
	\label{Proposition4}
		We first rewrite \textbf{P5.1} by introducing some auxiliary variables, which is given by 
	\begin{align}\tag{$\textbf{P.A}$}
		\max_{\bm W\succeq \bm 0} ~~& -\text{Tr}(\bm W) \\ 
		\text{s.t.}~~
		& {\text{A.1} }:	
		\frac{1}{2}\log_2(b_{1,k} +\sigma_{p,k}^2 )+\frac{1}{2}\log_2(b_{2,k}+\sigma_{p,k}^2 ) -\log_2(\text{Tr}(\bm {B}_{3,k} \bm W^{(x)})+\sigma_{p,k}^2 ) \nonumber\\
		&-\frac{b_{3,k}-\text{Tr}(\bm {B}_{3,k}\bm W^{(x)})}{(\text{Tr}(\bm {B}_{3,k}\bm W^{(x)})+\sigma_{p,k}^2)\ln2} \ge R_{\text{u,min}},~ k \in \mathcal{K}, \nonumber \\ 	
		& {\text{A.2} }:	
		\frac{1}{2}\log_2(\overline{b}_{1,q,k}+\sigma_{s,q}^2 )+\frac{1}{2}\log_2(\overline{b}_{2,q,k}+\sigma_{s,q}^2 ) -\log_2(\text{Tr}(\overline{\bm {B}}_{3,q,k} \bm W^{(x)})+\sigma_{s,q}^2 ) \nonumber \\
		& -\frac{\overline{b}_{3,q,k}-\text{Tr}(\overline{\bm {B}}_{3,q,k}\bm W^{(x)})}{(\text{Tr}(\overline{\bm {B}}_{3,q,k}\bm W^{(x)})+\sigma_{s,q}^2)\ln2}
		\ge R_{\text{u,min}},~ k \in \mathcal{K}, ~q \in \mathcal{Q}, \nonumber \\
		&{\text{A.3} }:
		\text{Tr}(\bm {B}_{u,k}\bm W)\ge b_{u,k}, u\in \{1,2,3\}, k\in \mathcal{K}, \nonumber \\
		&{\text{A.4} }:
		\text{Tr}(\overline{\bm {B}}_{u,q,k} \bm W)\ge \overline{b}_{u,q,k}, u\in \{1,2,3\}, ~ k \in \mathcal{K}, ~q \in \mathcal{Q},\nonumber \\
		& {\text{A.5} }:
		\text{Tr}(\overline{\bm G}_q \bm W)\ge \sigma_{s,q}^2 \gamma_{u,\text{min}},~q \in \mathcal{Q},
	\end{align}
where $\overline{\bm G}_q=L \overline{\bm R}_q -\gamma_{u,\text{min}} \overline{\bm H}_q$. The Lagrangian function of \textbf{P.A} about $\bm W$ is formulated as 
	$\mathcal{L}(\bm W) =\tilde{\mu} -\text{Tr}(\bm W) + \sum_{u=1}^{3}\sum_{k=1}^{K}\varepsilon_{u,k} \text{Tr}(\bm {B}_{u,k} \bm W) +\sum_{u=1}^{3}\sum_{q=1}^{Q}\sum_{k=1}^{K}\overline{\varepsilon} _{u,q,k} \text{Tr}(\overline{\bm {B}}_{u,q,k} \bm W)+ \sum_{q=1}^{Q}\tilde{\varepsilon}_q  \text{Tr} (\overline{ \bm G}_q \bm W)+ \text{Tr} (\bm \Omega \bm W)$,
where $\tilde{\mu}$ represents the term unrelated to $\bm W$, $\varepsilon_{u,k}$, $ \overline{\varepsilon} _{u,q,k} $, and $\tilde{\varepsilon} _{q}$ are the Lagrangian multipliers for the constraints A.3, A.4, and A.5, respectively. $\bm \Omega \succeq 0$  is the Lagrangian multiplier matrix related to $\bm W \succeq 0$. According to the KKT conditions of \textbf{P.A}, we can obtain
 \begin{align}
	\label{KKT1}
   &-\bm {I}_{NK \times NK} + \bm \Pi + \bm {\Omega}^*= \bm 0,\\
   \label{KKT2}
   & \text{Tr}(\bm {\Omega}^* \bm W^*) =0,
\end{align}
where $ \bm \Pi = \sum_{u=1}^{3}\sum_{k=1}^{K}( \varepsilon_{u,k}^* \bm {B}_{u,k}+ \sum_{q=1}^{Q}\overline{\varepsilon} _{u,q,k}^* \overline{\bm {B}}_{u,q,k})+\tilde{\varepsilon} _{q}^* \overline{ \bm G}_q $, $  \varepsilon_{u,k}^* $, 	$ \overline{\varepsilon} _{u,q,k}^* $, $ \tilde{\varepsilon} _{q}^* $, and $ \bm {\Omega}^* $ are the optimal Lagrangian multipliers. It is easy to find that $\bm \Pi$ is a block diagonal matrix, which is the summation of the block diagonal matrices $\bm {B}_{u,k} $, $ \overline{\bm {B}}_{u,q,k} $, and $\overline{ \bm G}_q$. Thus, we can formulate $\bm \Pi$ as
\begin{align}
	\bm{\Pi} = \text{BlDiag}\left(\{\bm{\Pi}_k\}_{k=1}^K\right),	
\end{align}
where $\bm {\Pi}_k \in \mathbb{C}^{N \times N},~k \in \mathcal{K}$. In order to satisfy \eqref{KKT1}, $\bm \Omega^* $ should also be a block diagonal matrix, which  is formulated as
\begin{align}
	\label{Omega}
	\bm{\Omega^*} = \text{BlDiag}\left(\{\bm{\Omega}_k\}_{k=1}^K\right),	
\end{align}
where $\bm {\Omega}_k^*=\bm I_{N \times N} - \bm {\Pi}_k $. According to \eqref{KKT2} and \eqref{Omega}, $\bm W^*$ can be expressed as
	\begin{align}
		\label{W}
		\bm{W}^{*} = \text{BlDiag}(\{\bm{W}_k^*\}_{k=1}^K).
	\end{align}
	Then, we have  $\sum_{k=1}^{K} \text{Tr}(\bm {\Omega}_k^* \bm W_k^* )=0$, where 
\begin{align}
\label{OmegaSub}
\text{Tr}(\bm {\Omega}_k^* \bm W_k^* )=0,~k \in \mathcal{K}.
\end{align}

According to \eqref{OmegaSub}, $\bm {\Omega}_k^*$  lies in the null space of $\bm W_k^*$. Hence,  we analyze the rank of $\bm W_k^*$ by investigating the characteristic of $\bm {\Omega}_k^*$ (i.e., $ \bm I_{N \times N} - \bm {\Pi}_k $ ), whose  eigenvalues are nonnegative. We denote $ \vartheta_{\text{max},k} $  as the largest eigenvalue of $\bm {\Pi}_k$. If $ \vartheta_{\text{max},k}<1 $, we find that all eigenvalues of $ \bm I_{N \times N} - \bm {\Pi}_k $ are positive and $\text{Rank}( \bm I_{N \times N} - \bm {\Pi}_k )=\text{Rank}(\bm{\Omega}_k^* )=N$. Since \eqref{OmegaSub} holds  if and only if $\bm W_k^*=\bm 0$, which is obviously not the optimal solution. If $ \vartheta_{\text{max},k}>1 $, the smallest eigenvalue of $ \bm {\Omega}_k^* $ is negative, which is against the fact that $ \bm {\Omega}_k^* \succeq \bm 0 $. If  $ \vartheta_{\text{max},k}=1 $, it is found that the eigenvector of the largest eigenvalue of $ \bm {\Pi}_k$ lies in the null space of $ \bm {\Omega}_k^* $, denoted by $\bm \omega_{\text{max},k}$. Thus, we obtain $\bm {W}_k^* = \varpi \bm \omega_{\text{max},k} \bm \omega_{\text{max},k}^H$, which is a rank-one matrix, and $\varpi >0$ is a scaling parameter. Since $\bm{W}^{*}$ is a block diagonal matrix with $\bm {W}_k^*$ for $k\in \mathcal{K}$ being its diagonal sub-matrices, we  derive that the rank of $\bm W$ is $K$, which completes the proof.

\end{document}